\documentclass[twocolumn,journal]{IEEEtran}
\pdfoutput=1

\usepackage{fancyhdr}
\usepackage{hyperref}
\usepackage{amsmath, amsthm, amssymb, graphicx}
\usepackage{xspace}
\usepackage{braket}
\usepackage[font={small}]{caption}
\usepackage{color}
\usepackage{setspace}
\usepackage{enumerate}
\usepackage{graphics}
\usepackage{cite}
\usepackage{dsfont}
\usepackage{latexsym}
\usepackage{amsmath}
\usepackage{amsfonts}
\usepackage{amssymb}
\usepackage{times}
\usepackage{sgame}
\usepackage{color}
\usepackage{verbatim}

\usepackage{enumitem}
\newlist{inparaenum}{enumerate}{2}% allow two levels of nesting in an enumerate-like environment
\setlist[inparaenum]{nosep}% compact spacing for all nesting levels
\setlist[inparaenum,1]{label=\bfseries\arabic*.}% labels for top level
\setlist[inparaenum,2]{label=\arabic{inparaenumi}\emph{\alph*})}% labels for second level
\newcommand{\be}{\begin{equation}}
\newcommand{\ee}{\end{equation}}
\newcommand{\mbb}[1]{\mathbb{#1}}

\newcommand{\mcal}[1]{\mathcal{#1}}

\newcommand{\hs}{\hspace{-.5mm}}

\newcommand{\aee}{{\cal A}}
 
\newcommand{\gee}{{\cal G}}

\newcommand{\nee}{{\cal N}}

\newcommand{\argmax}[1]{\underset{#1}
{\operatorname{arg}\,\operatorname{max}}\;}
\newcommand{\argmin}[1]{\underset{#1}
{\operatorname{arg}\,\operatorname{min}}\;}

\newtheorem{theorem}{Theorem}[section]

\newtheorem{proposition}{Proposition}[section]
\newtheorem{lemma}{Lemma}[section]

\newtheorem{definition}{Definition}

\normalsize\title{\LARGE \bf
	The Impact of Complex and Informed Adversarial Behavior in Graphical Coordination Games
	\thanks{This research was supported by UCOP grant LFR-18-548175, ONR grant \#N00014-17-1-2060 and NSF grant \#ECCS-1638214. The material in this paper substantially extends the conference paper \cite{Canty2018} by providing complete proofs and novel results on dynamic policies. The current paper extends the results by fully characterizing dynamic adversarial influence.}}

\author{
	Keith Paarporn$^{\star}$, Brian Canty$^{\star}$, Philip N. Brown, Mahnoosh Alizadeh, and Jason R. Marden
	\thanks{K. Paarporn, M. Alizadeh, and J.R. Marden are with the Department of Electrical and Computer Engineering, University of California, Santa Barbara. B. Canty is with CACI International. P. N. Brown is with the Department of Computer Science at the University of Colorado, Colorado Springs. Contact: \texttt{kpaarporn@ucsb.edu, \{alizadeh,jrmarden\}@ece.ucsb.edu, brian.canty@caci.com, philip.brown@uccs.edu}.
	
	*These authors contributed equally to this work.}
}

\begin{document}
\maketitle

\begin{abstract}

How does system-level information impact the ability of an adversary to degrade performance in a networked control system? How does the complexity of an adversary's strategy affect its ability to degrade performance? This paper focuses on these questions in the context of graphical coordination games where an adversary can influence a given fraction of the agents in the system, and the agents follow log-linear learning, a well-known distributed learning algorithm.  Focusing on a class of homogeneous ring graphs of various connectivity, we begin by demonstrating that minimally connected ring graphs are the most susceptible to adversarial influence. We then proceed to characterize how both (i) the sophistication of the attack strategies (static vs dynamic) and (ii) the informational awareness about the network structure can be leveraged by an adversary to degrade system performance. Focusing on the set of adversarial policies that induce stochastically stable states, our findings demonstrate that the relative importance between sophistication and information changes depending on the the influencing power of the adversary. In particular, sophistication far outweighs informational awareness with regards to degrading system-level damage when the adversary's influence power is relatively weak.  However, the opposite is true when an adversary's influence power is more substantial.

%While the results in this paper are adversarial-centric, these findings provide insight as to what system-level information should be obfuscated from potential adversarial actors to protect system behavior.   
\end{abstract}

%%%%%%%%%%%%%%%%%%%%%%%%%%%%%%%%%%%%%%%%%%%%%%%%%%%%%%%
\section{Introduction}
% !TEX root = main.tex

A networked system can be viewed as a collection of subsystems, each required to make local and independent decisions in response to available information.  The information available to each subsystem could pertain to local environmental conditions or the behavior of a selected group of neighboring agents in the system; hence, the information available to one subsystem could be vastly different than the information available to other subsystems. Regardless of the specific problem domain and informational characteristics, the underlying goal is to derive agent control policies that ensure the emergent collective behavior is desirable with respect to a system-level performance metric. 

A central focus of such systems is the design of networked control algorithms that provide strong guarantees on the quality of emergent outcomes.  A networked control algorithm can be viewed as a decision-making rule that specifies how subsystems respond to local conditions.  There are several noteworthy results in this domain ranging from  consensus and flocking \cite{Martinez2007,Olfati-Saber2007}, sensor allocation \cite{Zhu2009,Ramaswamy2016}, coordination of unmanned vehicles \cite{Jadbabaie2003}, and many others.  A common theme in all of these works is the following:  If all agents follow the prescribed decision-making rules, then the emergent behavior is both stable and desirable.  In contrast to this work, here we seek to address whether such decision-making rules are robust to adversarial interventions.

While the decentralization associated with distributed architectures is undoubtedly appealing for a host of reasons, it is important to highlight that this also introduces vulnerabilities. In particular, the decision-making process of individual subsystems can potentially be influenced by adversarial actors in the system through corrupting or augmenting the information to the subsystems.  Accordingly, in this paper we ask whether an adversary can exploit these interconnections to negatively influence the quality of the emergent collective behavior. Formal analysis of this interplay has emerged in recent years, often in the context of robust consensus, distributed optimization, and cyber-physical system security \cite{Pasqualetti_2013,Fawzi2014,Sundaram_2018}.

The focus of this paper is the susceptibility of a distributed algorithm known as \emph{log-linear learning} in networked control systems \cite{Wolpert2001,McKelvey1995,Blume1995}.  Log-linear learning has received significant attention recently in the area of distributed control, as it can often be employed to ensure that the resulting behavior is near optimal.  A representative set of examples range from control of wind farms \cite{Marden2013}, sensor networks \cite{Zhu2009,Ramaswamy2016,Lim2013,Rahili_2014,Sun_2018},  task assignment \cite{Kanakia_2016}, among others \cite{Shah2010}.  However, the susceptibility of this approach to adversarial interventions is generally unknown.

% Many of these references are irrelevant.  Need to do better literature review.  

%The baseline control strategy that we consider in this paper originates from the game theoretic literature on distributed control \cite{Li2013,Marden2014c,Marden2014,Pavel2012}. Here, the underlying approach to control networked multi-agent systems is to  (i) assign each agent a local objective function that is equal to the agent's marginal contribution to the true system-level objective and (ii) assign each agent a probabilistic distributed learning rule such as \emph{log-linear learning} \cite{Wolpert2001,McKelvey1995,Blume1995}. The allure of this approach is  it guarantees  that the emergent collective behavior  optimizes the system-level objective (in an asymptotic sense) for a broad class of multi-agent systems~\cite{Alos-Ferrer2010}. 

The goal of this paper is to shed light on the susceptibility of log-linear learning to adversarial interventions. To that end, we focus on a well-studied class of systems known as \emph{graphical coordination games} \cite{Kearns2001,Montanari2010}.  Graphical coordination games model strategic scenarios where agents are tasked with adopting conventions and derive benefits from coordinating with the choices of their neighbors, e.g., adoption of technology or conventions \cite{Young_2001,Montanari2010}. Regardless of the specifics of the graphical coordination game, log-linear learning is known to asymptotically achieve optimal system-level behavior. In this work, we focus on characterizing the degree to which the performance guarantees of log-linear learning algorithms can be undermined by adversarial manipulations. 

In particular, our goal is to evaluate how different adversarial features  can inflict harm on the system. How much more of a threat is an adversary that knows the underlying network structure versus one that does not? An adversary that can dynamically alter its strategy versus one that can not? We begin by stating our model to ensure that our contributions are clear.

\subsection{Model: Graphical Coordination Games}

We consider the framework of graphical coordination games where there is a collection of agents $ \mcal{N} = \{1,2,\ldots,n\}$ enmeshed in an underlying undirected network $G = (\mcal{N},\mcal{E})$ where $\mcal{E} \subseteq \mcal{N} \times \mcal{N}$ defines the inter-agent interconnections. There are two different conventions, denoted by $x$ and $y$, and each agent $i \in N$ must decide between a set of conventions  $\aee_i \subseteq \{x,y\}$. Note that if $\aee_i =  \{y\}$, this means that agent $i$ is required to select convention $y$.  The benefit agent $i$ associates with a choice $x$ or $y$ depends on how many of its network neighbors $\mcal{N}_i = \{j\in\mcal{N} : (i,j) \in \mcal{E}\}$ have selected the same convention.  More formally, given a joint action profile $a = (a_1,\ldots,a_n) \in \aee :=  \aee_1 \times \dots \aee_n$, the total benefit agent $i$ experiences is given by
\begin{equation}\label{eq:i_benefit}
    U_i(a) := \sum_{j \in \nee_i} V(a_i, a_j).
\end{equation}
where $V: \{x,y\}^2 \rightarrow \mbb{R}$ defines the per agent benefit of coordinating with a neighboring agent on a given convention.  Throughout, we consider $V$ of the following form
\begin{center}
\begin{game}{2}{2}
	& $x$ & $y$  \\
	$x$   &$1+\alpha, 1+\alpha $ & $0,0$\\
	$y$  &$0,0$ & $1,1$ \\
\end{game}%\hspace*{20mm}%
\end{center}
where $\alpha > 0$.  The \emph{system welfare} associated with the action profile  $a \in \aee$ is given by
\begin{equation} \label{eq:W}
    W(a) := \sum_{i \in \mcal{N}} U_i(a).
\end{equation}

The goal of a system operator is to assign decision-making rules for the agents such that their emergent collective behavior optimizes the system welfare, i.e., the emergent action profile is of the form
\begin{equation}
    a^{\rm opt} \in \underset{a \in \aee} {\arg \max} \ W(a).
\end{equation} 
One such algorithm that achieves this objective is \emph{log-linear learning} \cite{Marden2014,Lim2013,Marden2009,McKelvey1995,Pradelski2012,Alos-Ferrer2010}.  
Log-linear learning is a stochastic distributed algorithm that governs the evolution of agents' decisions over time. More formally, log-linear learning produces a sequence of joint action profiles $\{a(t)\}_{t= 0}^\infty$, which we also call states, determined by the following process: 

\begin{definition}[Log-Linear Learning]
Let $a(0)\in\mcal{A}$ be any action profile.  At each time $t \geq 1$, one agent $i \in \mcal{N}$ is selected uniformly at random and allowed to alter its action choice.  All other agents are required to repeat their previous action, i.e., $a_{-i}(t) = a_{-i}(t-1)$ where $a_{-i} = \{a_1, \dots, a_{i-1},a_{i+1}, \dots, a_n\}$ captures the action choice of all agents $\neq i$.  The updating agent $i$ selects any action $a_i \in \aee_i$ at time $t$ with probability 
	\begin{equation}
	    \frac{e^{\beta U_i(a_i, a_{-i}(t-1))}}{\sum_{\tilde{a}_i \in \aee_i} e^{\beta U_i(\tilde{a}_i, a_{-i}(t-1))}}
	\end{equation}
	where $\beta > 0$ is a given algorithm parameter.  Once agent $i$ selects her action, the process is repeated. 
\end{definition}

Log-linear learning induces in an ergodic process over the joint action profiles $\aee$ in any graphical coordination game of the above form. The stochastically stable states, which we express by $\text{LLL}(G, \aee, \alpha) \subseteq \aee$ is defined as the support of the limiting distribution as $\beta \rightarrow \infty$. In the context of graphical coordination games, log-linear learning ensures that 
\begin{equation}
    \text{LLL}(G, \aee, \alpha) = \underset{a \in \aee} {\arg \max} \ W(a).
\end{equation}
Note that log-linear learning guarantees that the emergent behavior optimizes the system-level objective irrespective of the graph $G$, the convention choices available to the agents $\aee$, and the value of $\alpha$.  Note that in the special case when $\aee_i = \{ x, y\}$ for all $i \in N$, then $\text{LLL}(G, \aee, \alpha) = \{ \vec{x} = \{x, \dots, x \}\}$ is the all $x$ convention.  For alternative choices of $\aee$, the action profiles that optimize system welfare is not as straightforward.  

\subsection{Models of Adversarial Interventions}

In this paper we consider an adversary seeking to influence the decision-making process of log-linear learning by strategically integrating $S=\{1,\dots, |S|\}$ adversarial nodes into the system. Each of the adversarial nodes $s \in S$ will be integrated into the network through a connection to a unique single agent $i \in \nee$ that the adversarial node is tasked with influencing though a choice $a_s = \{x\}$ or $a_s = \{y\}$.  Let $S_x, S_y \subseteq \nee$, $|S_x|+|S_y| \leq |S|$, denote the set of agents that are being influenced by an adversary promoting $\{x\}$ and $\{y\}$ respectively.  Given $S_x$ and $S_y$, the influenced utility of an agent $i \in \nee$ is of the form 
\begin{equation}\label{eq:perceived_u}
\tilde{U}_i(a;S_x,S_y)\!:=\!\left\{\hs\hs\hs\begin{array}{lcl}
U_i(a) + V(a_i, x) & & \hs\text{if} \ i \in S_x \\
U_i(a) + V(a_i, y) & & \hs\text{if} \ i \in S_y \\
U_i(a)  & & \hs\text{else}
\end{array}\right.
\end{equation}
In words, an agent $i\in S_y$ (resp. $i\in S_x$) experiences the usual benefits from its neighbors in $\nee_i$, plus an additional utility of $1$ if $a_i=y$ (resp. $1+\alpha$ if $a_i=x$). While the adversarial nodes $S$ do not directly contribute to the system-level objective as defined in \eqref{eq:W}, they modify the network agents' utility functions, which invariably influence the resulting asymptotic behavior associated with log-linear learning. We now denote by $\text{LLL}(G, \aee, \alpha, \pi)$ the (possibly modified) stochastically stable states, where $\pi$ defines the process, or \emph{policy}, through which $S_x$ and $S_y$ are chosen. Technically speaking, the sets $S_x(t),S_y(t)$ are drawn from the distribution $\pi(t)$. The performance degradation associated with the adversarial policy $\pi$ is measured by
\begin{equation}\label{eq:eff}
    \eta(G,\aee,\alpha,\pi): = \min_{a \in \text{LLL}(G, \aee, \alpha, \pi)} \left\{\frac{ W(a)}{W(a^{\rm opt})} \right \} \geq 0.
\end{equation}
We will focus on graphical coordination games where the agents have full choice of conventions, i.e., $\aee_i = \{x,y\}$ for all $i \in N$.  For that setting, we will omit highlighting the dependence of $\aee$ in the definition of $\eta(\cdot)$ and $\text{LLL}(\cdot)$, i.e. we will instead write $\eta(G,\alpha,\pi)$ and $\text{LLL}(G,\alpha,\pi)$.

\subsection{Summary of Contributions}

The focus of this manuscript is on characterizing the susceptibility of log-learning learning to adversarial interventions in networked coordination games.  In particular, our goal is to identify the salient features of the worst-case adversarial policies.  Specifically, we focus on identifying the importance of the following two attributes: 

    \noindent\textbf{-} \textbf{Informational Awareness:} Does the adversary know the network structure? 
    
    \noindent\textbf{-} \textbf{Strategic Sophistication:} Can the adversarial nodes dynamically alter their location and convention choice over time?

The above attributes define four classes of adversarial policies, which we represent by $\{\Pi_{\text{I,D}},\Pi_{\text{I}},\Pi_{\text{D}},\Pi\}$, where the subscript $I$ denotes informationally aware and the subscript $D$ denotes dynamic adversarial policies. The absence of a subscript distinction means the negation.  For example, $\Pi_{\text{I,D}}$ denotes the set of adversarial policies that are dynamic and can utilize information about the network structure.  On the other hand, $\Pi$ denotes the set of adversarial policies that are static and agnostic to network structure.  By dynamic, we mean that the adversary can alter its behavior based on the current network state. That is, we consider stationary policies\footnote{In this paper, we restrict attention to adversarial policies that induce stochastically stable states -- in particular, static policies and dynamic policies that are stationary. It will be of interest in future work to investigate other types of dynamic policies that may not guarantee a SSS is induced.} $\{S_x(a(t)), S_y(a(t))\}_{t = 1, 2, \dots}$. A static policy does not allow this flexibility: $S_x(a(t)) = S_x$ and $S_y(a(t)) = S_y$ $\forall t$.

Our first set of main results identify the most vulnerable graph structures. Focusing on a class of  homogeneous ring graphs where an adversary can influence at most $\gamma \cdot n$ agents, where $\gamma \in [0,1]$, we demonstrate that the most susceptible, i.e.,  graphs that lead to the lowest efficiency as defined in \eqref{eq:eff}, are minimally connected ring graphs. We demonstrate this over the set of adversarial policies $\Pi$ and $\Pi_{\text{D}}$ (Theorems \ref{theorem:SU_opt} and \ref{theorem:DU_opt}).  This matches intuition as the graph with the fewest internal edges are in fact the most susceptible to adversarial interference.

Our second set of results focus exclusively on these ring graphs and seek to identify how information regarding the network structure can be exploited by the adversary.  In doing so, we characterize the tight worst-case performance guarantees as in \eqref{eq:eff} over policies belonging to $\Pi_{\text{I}}$ and $\Pi_{\text{I,D}}$ (Theorems \ref{theorem:SI_opt} and \ref{theorem:MI_opt}).  Figure~\ref{fig:plots} highlights an instance of worst-case performance guarantees for all four types of policies $\{\Pi_{\text{I,D}},\Pi_{\text{I}},\Pi_{\text{D}},\Pi\}$ when $\alpha = 0.5$.  As expected, the adversary leverages information and sophistication to most effectively degrade performance guarantees.  However, the regimes where each of these attributes is most valuable is not so predictable.  When an adversary has limited strength, i.e., $\gamma < 0.5$, sophistication is far more valuable than informational awareness to the adversary.  That is, the best adversarial policy in $\Pi_D$ significantly outperforms the best adversarial policy in $\Pi_I$.  When an adversary has more substantial strength, i.e., $\gamma > 0.5$, the opposite is true. We formalize these conclusions in Theorem \ref{theorem:comparison}. Theorem \ref{theorem:sat} highlights the performance differences between static and dynamic policies. In particular, dynamic policies can achieve the same performance as static ones using fewer adversarial nodes, but such performance saturates above a threshold budget. 

We provide proofs in Sections \ref{sec:static_analysis} (static adversaries), \ref{sec:dynamic_analysis} (dynamic adversaries), and the Appendix (Theorems \ref{theorem:comparison} and \ref{theorem:sat}).

%That is, the best adversarial policy in $\Pi_I$ significantly outperforms the best adversarial policy in $\Pi_D$.

\begin{figure}[t]
	\centering
	\includegraphics[scale=.34]{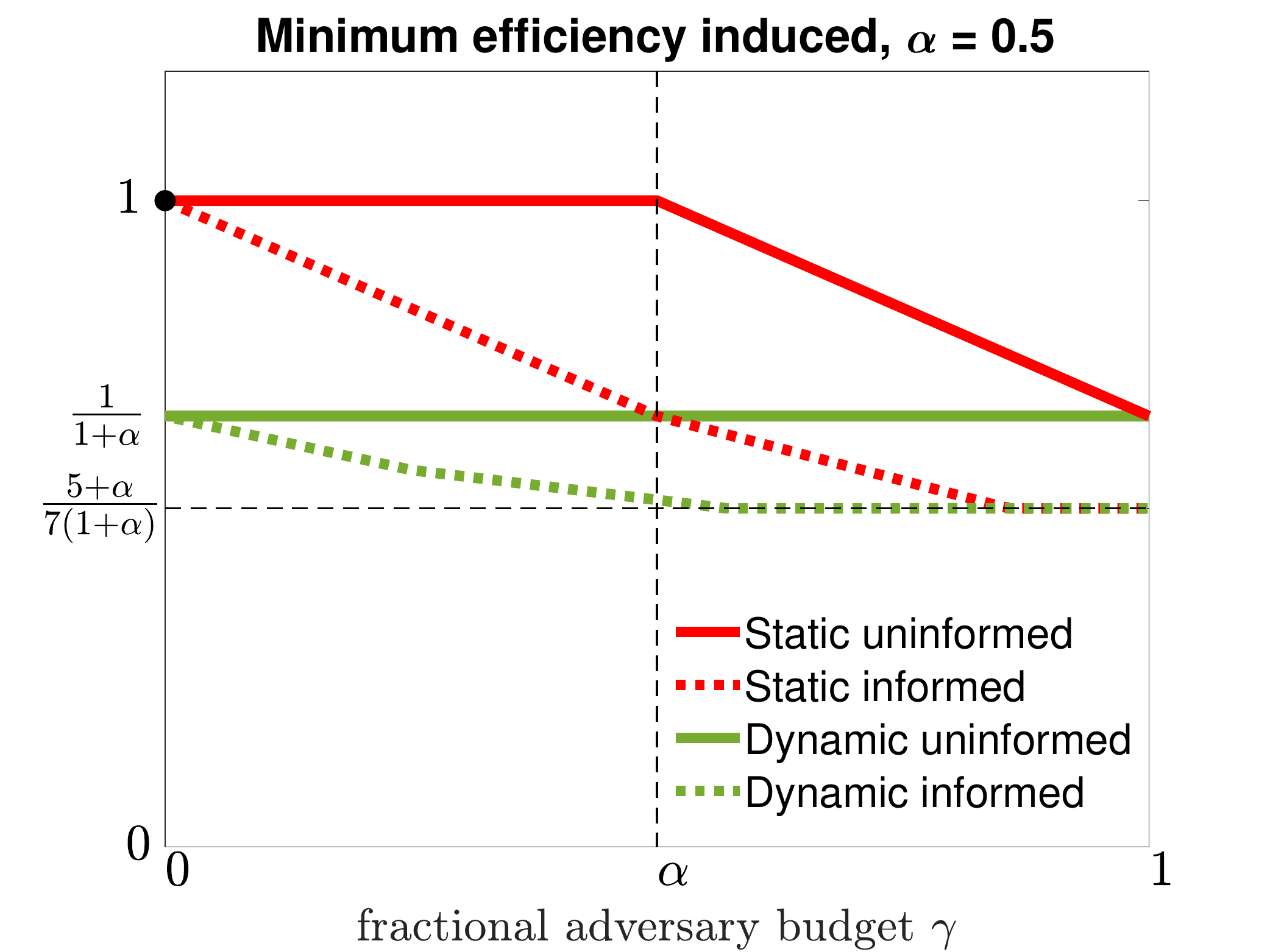}
	\caption{ \small  This figure highlights the interplay between an adversary's informational awareness (informed vs uninformed), strategic sophistication (static vs dynamic), budget, and the minimum efficiency it can induce on the system. The green and red lines characterize minimum efficiencies induced from four different adversarial models on ring networks of sufficiently large size, as a function of fractional budget $\gamma \in [0,1]$ (the fraction of agents the adversary can influence). At $\gamma = 0$, neither adversarial model can induce any damage on efficiency (black circle).  For low budgets (i.e. $\gamma < 0.5$), strategic sophistication is more valuable than having system-level information about the network. The converse holds true for higher budgets (i.e. $\gamma > 0.5$): system-level information is more valuable than the ability to implement dynamic policies.    }
	\label{fig:plots}
	% The performance of dynamic informed adversaries (dotted green line) saturates at a certain level $\gamma \approx 0.5$, while a static informed adversary's performance saturates to the same level at $\gamma \approx 0.86$.
\end{figure}

\subsection{Related Work}

Previous work has studied to what extent networked distributed algorithms, designed to operate in the absence of adversarial interference, are susceptible to such influence \cite{Lamport_1982,Pasqualetti_2012,LeBlanc_2013,Sundaram_2018,Jaleel_2019,Su_2016}. For example, distributed multi-agent optimization algorithms are shown to easily be compromised by adversarial behaviors \cite{Su_2016,Sundaram_2018}. Indeed, there are fundamental limitations to these algorithms and their variants. Such algorithms cannot  perform optimally in the absence of adversaries as well as be resilient to adversarial attacks at the same time \cite{Sundaram_2018,Su_2016,Paarporn_2019CDC}.

How emergent behavior associated with game-theoretic learning algorithms, such as log-linear learning, could be influenced by adversarial nodes was initially studied in \cite{Borowski2015,Brown2018TCNS}. The focus centered on how easily adversarial nodes could steer agents towards an inefficient Nash equilibrium. In this paper, we instead focus on an adversary seeking to minimize system-level performance.

%%%%%%%%%%%%%%%%%%%%%%%%%%%%%%%%%%%%%%%%%%%%%%%%%%%%%%%
\section{Analysis of Susceptible Graphs}

% !TEX root = main.tex
This section focuses on identifying which graph structures are most susceptible to adversarial influence.  To that end, we will focus on a class of graphs that we term $k$-connected ring graphs, for $k \in \{1,\ldots,\lfloor n/2 \rfloor\}$. A graph $G=(\nee,\mcal{E})$ is a $k$-connected ring graph if $\nee_i = \{i-k, \dots, i-1, i+1, \dots, i+k\}$ for each agent $i \in \mathcal{N}$, where addition and subtraction are both modulo $n$.  Note that when $k=1$ we have the usual ring graph and when $k = \lfloor n/2 \rfloor$ we have the complete graph. Let us denote $\mcal{G}_n^k$ as the set of all $k$-connected ring graphs of size $n$. The following Theorems outline the degradation in performance attainable through admissible adversarial policies belonging to $\Pi$ and $\Pi_{\text{D}}$ -- static and dynamic uninformed adversaries, respectively.

\begin{theorem}\label{theorem:SU_opt}
    Consider the class of network coordination games where (i)\footnote{Values of $\alpha \geq 1$ are not considered here. If $\alpha \geq 1$, then a single $x$ link is valued as much or higher than two $y$ links. If this is the case, no $y$ agents can be induced in the stochastically stable state under any adversarial policy on ring graphs, and no damage can be inflicted.} $\alpha \in [0,1)$, and (ii) an admissible adversarial policy can influence at most a fraction $\gamma \in [0,1]$ of agents in the network. Recall $\Pi(G,\gamma)$ is the set of admissible static adversarial policies that are agnostic about the network structure. Then,
	\begin{equation}\label{eq:SU_OPT}
	    \lim_{n\rightarrow \infty}\inf_{\substack{\pi \in \Pi(G,\gamma) \\ G \in \mathcal{G}_n^k}} \!\! \eta(G,\alpha,\pi) =
		%\eta^k(\alpha,\gamma) =
		\begin{cases}
			1, \!\!&\text{if } \gamma < k\alpha \\
			\frac{(1 - (k-1)\alpha) - \alpha \gamma}{(1+\alpha)(1-k\alpha)}, &\text{if } \gamma \geq k\alpha
		\end{cases}.
	\end{equation}
	%{\color{blue}Furthermore, the limit of efficiency as the size of $G$ grows ($n\rightarrow \infty$) equals the lower bound. }
\end{theorem}
\begin{theorem}\label{theorem:DU_opt}
    Consider $\Pi_{\text{D}}(G,\gamma)$, the set of admissible dynamic adversarial policies that are agnostic about the network structure on any graph $G \in \mathcal{G}_n^k$. Then for $\alpha \in [0,1)$ and  $\gamma \in [0,1]$,
	\begin{equation}\label{eq:DU_OPT}
		\inf_{\pi \in \Pi_{\text{D}}(G,\gamma)} \eta(G,\alpha,\pi) \geq
		%\eta^k_{D}(\alpha,\gamma) =
		\begin{cases}
			\frac{1}{1+\alpha}, &\text{if } \alpha < \frac{1}{k}, \gamma \neq 0  \\
			1, &\text{if } \alpha \geq \frac{1}{k} \text{ or } \gamma=0
		\end{cases}.
	\end{equation}
	Furthermore, the limit of efficiency as the size of $G$ grows ($n\rightarrow \infty$) equals the lower bound.
\end{theorem}

There are several interesting things to note from Theorems \ref{theorem:SU_opt} and \ref{theorem:DU_opt}.  If $\alpha \geq 1/k$, neither classes of adversarial policies $\Pi_{\text{D}}(G,\gamma)$ nor $\Pi(G,\gamma)$ can inflict any damage on the system regardless of the budget $\gamma$. Second, the achievable efficiency of a dynamic uninformed adversary, i.e., restriction to $\Pi_{\text{D}}(\gamma)$, is constant for $\gamma \in (0,1]$. Third, the induced efficiency from a static uninformed adversary, i.e., restriction to $\Pi(G,\gamma)$, is  decreasing in $\gamma$. Lastly, by tightness we know that for any $k \geq 1$ and $\gamma \in (0,1]$ we have 
\begin{align*}
		\inf_{G^1 \in \gee^1, \pi \in \Pi} \eta(G^1,\alpha,\pi) \leq
        \inf_{G^k \in \gee^k, \pi \in \Pi} \eta(G^k,\alpha,\pi),     
%        \inf_{G^1 \in \gee^1, \pi \in \Pi_{\text{D}}} \eta(G^1,\alpha,\pi) &\leq
%        \inf_{G^k \in \gee^k, \pi \in \Pi_{\text{D}}} \eta(G^k,\alpha,\pi),
\end{align*}
and an identical relation holds for policies in $\Pi_{\text{D}}$. Here, we omit highlighting the dependence on $\Pi(\cdot)$ for brevity. Hence, ring graphs ($k=1$) are the graphs that are most susceptible to adversarial interference.  %This holds true regardless of whether we consider the class of adversarial policies $\Pi$ or $\Pi_D$.

%%%%%%%%%%%%%%%%%%%%%%%%%%%%%%%%%%%%%%%%%%%%%%%%%%%%%%%
\section{The Impact of Information on Ring Graphs}

% !TEX root = main.tex

The previous section demonstrated that ring graphs are the most susceptible to adversarial influence.  In this section we explicitly characterize the impact informational awareness has on the potential degradation by admissible adversarial policies. We focus this analysis exclusively on ring graphs.   

\subsection{Static Informed Adversarial Policies}

This section focuses on the potential degradation of the adversarial policies in the set $\Pi_{\text{I}}$.  By knowing the graph's structure, the adversary can explicitly target specific agents $S_x, S_y \subseteq \nee$ in the network. An adversarial policy $\pi \in \Pi_{\text{I}}$ defines the process by which these agents are selected.  The resulting policies is static is the sense that for all times $t \geq 1$, $S_x(t), S_y(t) = S_x, S_y$.  The following Theorem characterizes the potential degradation caused by such adversarial policies.  

\begin{theorem}\label{theorem:SI_opt}
    Consider the class of network coordination games where (i) $\alpha \in [0,1)$ and (ii) $G \in \mcal{G}_n^1$. Given a fractional adversarial budget consider $\Pi_{\text{I}}(G,\gamma)$, the set of admissible adversarial policies that are static, but can depend on the network structure. Then for $\gamma \in (0,1]$,
    %
% 	\begin{equation}\label{eq:SI_bound}
% 		\inf_{\pi \in \Pi_{\text{I}}(G,\gamma)} \eta(G,\alpha,\pi) \geq \begin{cases} \eta_{\text{I}}(\alpha,\gamma), &\text{ if } \gamma \in (0,1] \\ 1, &\text{ if } \gamma = 0 \end{cases}
% 			%&\eta_I(\alpha,\gamma) =  \\
% 	\end{equation}
% 	where $\eta_{\text{SI}}(\alpha,\gamma)$ is the value of the following integer optimization problem:
	\begin{equation}\label{eq:SI_opt}
		\begin{aligned}
		    &\inf_{\pi \in \Pi_{\text{I}}(G,\gamma)} \eta(G,\alpha,\pi) \geq \\
			&\underset{\ell_{x_1},\ell_{x_2},\ell_{y_1},\ell_{y_2}\in\mathbb{Z}_{\geq 0}}{\inf}  \frac{1}{1+\alpha}\!\left(\!1\!+\!\frac{(2+\alpha)(\frac{s_1}{s_2}\!-\!1)+\alpha(\ell_{x_1}\!-\!\frac{s_1}{s_2}\ell_{x_2})}{\ell_{x_1}\!+\!\ell_{y_1}\!-\!\frac{s_1}{s_2}(\ell_{x_2}\!+\!\ell_{y_2})} \right), \\
			&\text{subject to: for } j = 1,2, \\
			&\ell_{x_j} \geq 2, \quad \ell_{y_j} \geq \left\lceil \frac{2+\alpha}{1-\alpha} \right\rceil  \\
			&s_{j}\!=\!\gamma (\ell_{x_j}\!+\!\ell_{y_j})\!-\!\left\lceil\alpha(\ell_{y_j}\!+\!1)\right\rceil\!-\!2 -\!\left\lceil\!\frac{\left[2-\alpha(\ell_{x_j}-1)\right]_+}{1+\alpha}\!\right\rceil \\
			&s_1 = 0 \text{ with } \ell_{x_2},\ell_{y_2} = 0,  \quad \text{or } s_1 > 0 \text{ and } s_2 < 0 
		\end{aligned}. \tag{SI-OPT}
	\end{equation}
	For $\gamma = 0$, the efficiency for any graph is 1. Here, we denote $[z]_+ = \max\{z,0\}$ for any $z\in\mbb{R}$. Furthermore, the limit of efficiency as the size of $G$ grows ($n\rightarrow \infty$) equals the lower bound \eqref{eq:SI_opt}.
\end{theorem}
%

%{\color{red}KEITH:  Is rational need just for the equality?  }

There are several interesting things to note from Theorem~\ref{theorem:SI_opt}, which characterizes the greatest damage  that an adversary can inflict upon the system when relying on static policies that can depend on the graph structure.  This theorem informs the structure of the worst-case attack, which involves the adversary attempting to stabilize alternating $x,y$ sequences of four distinct lengths.  While the structure of this adversarial attack is not necessarily fundamental, the interesting part of the theorem centers on tightness. That is, the adversary can never inflict more damage than the bounds given in Theorem~\ref{theorem:SI_opt}, and the best adversarial strategy approaches this bound as the size of the ring graph in consideration gets larger.

%and there is a ring graph of a specific size $n>1$ where the best adversarial strategy achieves this bound. 

% For low budgets $\gamma < \alpha$, one can express this bound in a more convenient form.
% %
% \begin{corollary}
%     Suppose $\alpha \in (0,1]$ and $\gamma < \alpha$. Then,
%     \begin{equation}
%         \inf_{G \in \mcal{G}^1, \pi \in \Pi_{\text{I}}(G,\gamma)} \eta(G,\alpha,\pi) = 1 - \frac{1}{1+\alpha}\gamma.
%     \end{equation}
% \end{corollary}
% %
% Hence in this regime, there is a linear relation between adversary budget and the damage inflicted.

\subsection{Dynamic Informed Adversarial Policies}

This section focuses on the potential degradation of the adversarial policies in the set $\Pi_{DI}$.  Here, the adversary can target specific agents $S_x(a(t)), S_y(a(t)) \subseteq \nee$  using knowledge of the graph structure $G$ and the sequence of action profiles $\{a(t)\}_{t\in\mathbb{Z}_{\geq 0}}$.   An adversarial policy $\pi \in \Pi_{DI}$ defines the process by which these agents are selected.  The following Theorem characterizes the maximum potential degradation caused by such adversarial policies.

\begin{theorem}\label{theorem:MI_opt}
    Consider $\Pi_{DI}(G,\gamma)$, the set of admissible adversarial policies that are dynamic and can depend on network structure, and $\alpha \in [0,1)$. Then the fundamental lower bound for $\inf_{\pi\in\Pi_{\text{DI}}(G,\gamma)} \eta(G,\alpha,\pi)$ is given by the RHS of \eqref{eq:SI_opt}, where the $s_j$ variables are instead
    \begin{equation}
        s_j =
        \begin{cases} 
            \gamma (\ell_{x_j} + \ell_{y_j})-4 \quad \text{if } \alpha<\frac12 \text{ and } \ell_{x_j} \leq 1 + \left\lfloor \frac{1-\alpha}{\alpha} \right\rfloor  \\
            \gamma (\ell_{x_j} + \ell_{y_j})-2\quad \text{else}
        \end{cases}
    \end{equation}
    for $j = 1,2$. Furthermore, the limit of efficiency as the size of $G$ grows ($n\rightarrow \infty$) equals the lower bound. 
\end{theorem}

%There are several interesting things to note from Theorem~\ref{theorem:MI_opt}, which characterizes the greatest damage  that an adversary can inflict upon the system when relying on dynamic policies that can depend on the graph structure and previous action profiles.   
Similar to \eqref{eq:SI_opt}, the lower bound of Theorem \ref{theorem:MI_opt} takes the form of an integer programming problem. While the structure of this adversarial attack is not necessarily fundamental, the interesting part of the theorem centers on tightness. That is, the adversary can never inflict more damage than the bound described in Theorem \ref{theorem:MI_opt} and the best adversarial strategy approaches this bound as the size of the ring graph gets larger.

\subsection{Comparison Between Information and Sophistication}
Here, we emphasize the qualitative differences between information and sophistication. The Theorem below asserts that sophistication, i.e. the ability to implement a dynamic policy, is a more desirable attribute for the adversary if its budget is relatively low, while information is more valuable if its budget is high.
\begin{theorem}\label{theorem:comparison}
%
%In low gamma regime, static informed is worse than dynamic uninformed.  
%In high gamma regime, static informed is better than dynamic uninformed.
%(would be nice to characterize non-trivial regime.  
%
Suppose $\alpha \in [0,1)$. For budgets $\gamma \in (0,\alpha)$ (empty interval if $\alpha=0$), we have
    \begin{equation}
        \lim_{n\rightarrow \infty} \inf_{\substack{G \in \mcal{G}_n^1 \\ \pi \in \Pi_{\text{D} }(G,\gamma)}} \eta(G,\alpha,\pi) < \lim_{n\rightarrow \infty}\inf_{\substack{G \in \mcal{G}_n^1 \\ \pi \in \Pi_{\text{I}}(G,\gamma)}} \eta(G,\alpha,\pi).
    \end{equation}
For budgets $\gamma \in (\alpha,1]$, the opposite (strict) inequality holds. They are equal if $\gamma = \alpha$.
\end{theorem}
Hence, in the low budget regime $\gamma < \alpha$, the adversary prefers to be uninformed and dynamic over being informed but static. The opposite conclusion holds in the high budget regime $\gamma > \alpha$. This characterization allows us to explicitly identify the importance of information and sophistication in adversarial policies as highlighted in Figure~\ref{fig:plots}\footnote{Figure 1 plots the bounds that the four main results, Theorems \ref{theorem:SU_opt}, \ref{theorem:DU_opt}, \ref{theorem:SI_opt}, and \ref{theorem:MI_opt}, characterize. While the bounds are analytically derived for Theorems \ref{theorem:SU_opt}, \ref{theorem:DU_opt} (uninformed adversaries), the plots for informed adversaries resemble a closely approximated value by solving their respective integer optimization problems with a finite upper bound of 100 on the decision variables.}.

The next result provides a comparison between static and dynamic informed adversaries. It states that given a sufficiently large adversarial budget, an optimal static informed policy can do just as much damage as an optimal dynamic informed policy.
\begin{theorem}\label{theorem:sat}
The fundamental lower bound on performance for static informed policies is
\begin{equation}\label{eq:sat_val}
	\left(\frac{1}{1+\alpha}\right)\frac{\ell^*+\alpha}{\ell^* + 2}
\end{equation}
if and only if it has a budget $\gamma \geq \gamma_{\text{sat}}^{\text{SI}} := \frac{\ell^* + \left\lceil \frac{2 - \alpha}{1+\alpha}\right\rceil}{\ell^* + 2}$, where $\ell^* := \left\lceil \frac{2 + \alpha}{1-\alpha}\right\rceil$. Furthermore, the fundamental lower bound on performance for dynamic informed policies coincides with \eqref{eq:sat_val} for budgets $\gamma \geq \gamma_{\rm sat}^{\text{DI}} := \frac{2 + 2\cdot\mathds{1}\left(\alpha < \frac{1}{2} \right)}{\ell^* + 2}$, where $\mathds{1}(\cdot)$ is the indicator function.
% $\ell^* := \min\{\ell \in \mbb{Z}_{\geq 0} : \ell = \lceil \alpha(\ell + 1) \rceil + 2 \}$
\end{theorem}
In other words, there are saturation levels on budget for both types of adversaries (DI and SI), where influencing more than $\gamma_{\text{sat}}^{\text{DI}}$ ($\gamma_{\text{sat}}^{\text{SI}}$) fraction of agents does not offer any additional performance gains. However, a static adversary will not exhibit saturation if $\alpha < \frac{1}{2}$. That is, the static adversary achieves performance level \eqref{eq:sat_val} if and only if it has a full budget $\gamma = 1$. It is interesting to note from the above Theorem that the dynamic informed adversary can maintain the performance level \eqref{eq:sat_val} for a wider range of budgets $\gamma \in [\gamma_{\text{sat}}^{\text{DI}},1] \supseteq [\gamma_{\text{sat}}^{\text{SI}},1]$ than the static informed adversary can. Here, the range is the same (no saturation exhibited for either) if and only if $\alpha = 0$.  Essentially, dynamic policies can inflict the same level of damage with fewer adversaries than a static policy. The proofs of both Theorems in this subsection are given in the Appendix.

%%%%%%%%%%%%%%%%%%%%%%%%%%%%%%%%%%%%%%%%%%%%%%%%%%%%%%%
\section{Proofs: Performance of static policies} \label{sec:static_analysis}
% !TEX root = main.tex

In this section, we provide proofs for the minimum efficiency a static adversary can induce. We will first prove Theorem \ref{theorem:SI_opt}, the case of a static informed adversary. As discussed, we limit our attention here to ring graphs $G \in \mcal{G}^1$. We then give a proof of Theorem \ref{theorem:SU_opt}, the case of a static uninformed adversary. This result relies on extending an intermediate step from the proof of Theorem \ref{theorem:SI_opt} to $k$-connected ring graphs.

The adversary's objective is to steer the system to a stochastically stable state of minimal efficiency. We will refer to action profiles that can be stabilized through some static policy  as the set of \emph{target profiles} a static uninformed and static informed adversary can induce, respectively. Indeed, we would like to characterize the target profile of minimal efficiency an adversary can achieve over any ring graph, i.e. 
\begin{equation}\label{eq:SI_problem}
    \inf_{G \in \mcal{G}^1, \pi \in \Pi_{\text{I}}(G,\gamma)} \eta(G,\alpha,\pi).
\end{equation}
Our approach is to view any action profile $a$ (and hence any target profile) as composed of alternating $x$ and $y$ segments. A $y$ segment $L_y$ is any subset $\{j,j+1,\ldots,j+|L_y|-1\} \subseteq \mcal{N}$ such that $a_i = y$ $\forall i \in L_y$ and $a_{j-1} = a_{j+|L_y|} = x$ (modulo $n$ arithmetic). Similarly, $L_x$ describes any such segment of $x$ agents.

\subsection{Proof of Theorem \ref{theorem:SI_opt}}

To begin, we start with a general outline of the forthcoming proof, which we break up into three steps. Following the outline, we give proofs for each of the individual steps.

\noindent\textbf{Step 1: Necessary and sufficient budget conditions to stabilize target profiles}

We derive the minimum number of adversarial nodes that is necessary and sufficient to stabilize a given action profile $a$. Indeed, suppose $S$ is an allocation of adversarial nodes. Then $a$ is stochastically stable if and only if for every $y$ segment $L_y$ and $x$ segment $L_x$ contained in $a$,
\begin{align}
    |S_y \cap L_y| &\geq \lceil \alpha(|L_y|+1) \rceil + 2 \label{eq:y_nec}, \\
    |S_x \cap L_x| &\geq \left \lceil  \frac{[2-\alpha(|L_x|-1)]_+}{1+\alpha}  \right \rceil \label{eq:x_nec},
\end{align}
and the spacing from two sequential $y$ adversarial nodes within $L_y$ is no more than $\left\lceil \frac{1}{\alpha} \right\rceil$. We observe that segment lengths must satisfy $|L_y| \geq \left\lceil\frac{2+\alpha}{1-\alpha}\right\rceil$ and $|L_x| \geq 2$.

\noindent\textbf{Step 2: Characterizing minimal efficiency target profiles}

Having established the number of adversarial nodes needed to stabilize target profiles, we identify structural properties of minimal efficiency target profiles that are stabilizable within the budget $\gamma \in (0,1]$. In particular, we show that
\begin{enumerate}
    % \item[(2A)] An adversary must utilize its full budget to induce maximal damage. Specifically, if the policy $\pi_S \in \Pi_{\text{I}}(G,\gamma)$ with $|S| < \lfloor \gamma \cdot n \rfloor$ stabilizes profile $a$, then one can always use a policy $\pi_{S'}$ with $|S| = \lfloor \gamma \cdot n \rfloor$ that also stabilizes $a$.
    \item[(2A)] Among adversarial policies that induce maximal damage, there is at least one that utilizes its full budget. Specifically, if the policy $\pi_S \in \Pi_{\text{I}}(G,\gamma)$ with $|S| < \lfloor \gamma \cdot n \rfloor$ stabilizes profile $a$, then one can always use a policy $\pi_{S'}$ with $|S| = \lfloor \gamma \cdot n \rfloor$ that also stabilizes $a$.
    \item[(2B)] The target profile of minimal efficiency contains at most two unique $x$ $y$ segment patterns.
\end{enumerate}

\noindent\textbf{Step 3: Optimization over worst-case target profiles}

We formulate an integer optimization problem whose solution gives \eqref{eq:SI_problem}. The decision variables are the lengths of the two unique $x$ $y$ segment patterns, subject to necessity constraints derived from \eqref{eq:y_nec} and \eqref{eq:x_nec}, as well as constraints given by the structural properties (2A) and (2B) of minimal efficiency target profiles. This formulation yields \eqref{eq:SI_opt}, and thus the proof of Theorem \ref{theorem:SI_opt}. \\

Before getting into the proofs of the claims given in the outline, we first present preliminary analytical tools for characterizing the emergent behavior when an  adversarial policy $\pi_S \in \Pi_{\text{I}}$  interferes with the agents' log-linear learning dynamics. Specifically, we seek to compute the stochastically stable states $\text{LLL}(G,\alpha,\pi_S)$. To do this, we can rely on the fact the graphical coordination game with static adversarial influence has a potential game structure \cite{Monderer1996}. In potential games, the stochastically stable states associated with log-linear learning are the action profiles that maximize the potential function \cite{Blume1993,Marden2012}. One can show that $\phi(a;S) := \frac{W(a)}{2} + \sum_{i\in S_x} V(a_i,x) + \sum_{i\in S_y} V(a_i,y)$. is a potential function for this game. Here, $\phi$ simply measures the number of coordinating links, including those induced from adversaries, weighted by their payoffs (i.e. $x$ or $y$ links).  Hence, for any graph $G$ and static policy $\pi_S$, we have $\text{LLL}(G,\alpha,\pi_S) = \argmax{a\in\mcal{A}} \phi(a; S)$.

%
% \begin{equation}\label{eq:phiSxy}
% 	\phi(a;S) := \frac{W(a)}{2} + \sum_{i\in S_x} V(a_i,x) + \sum_{i\in S_y} V(a_i,y).
% \end{equation}
%
%
% \begin{equation}\label{eq:potential_max}
% 	\text{LLL}(G,\alpha,\pi_S) = \argmax{a\in\mcal{A}} \phi(a; S)
% \end{equation}
%

\noindent\textbf{\underline{Proof of Step 1}}

We present the proof only for $y$ segments, as the arguments for $x$ segments are analogous. Suppose $a$ is stochastically stable, and contains a $y$ segment  $L_y = \{j,j+1,\ldots,j+|L_y|-1\}$. That is, $a_i = y$ for $i\in L_y$, and $a_{j-1} = a_{j+|L_y|} = x$. Consider any deviation $a'$ from $a$ that differs only within the segment $L_y$. Then it holds that $\phi(a';S) \leq \phi(a;S)$. In particular, if $a'$ is the profile where all agents in $L_y$ deviate to $x$, then $(1+\alpha)(|L_y|+1)  \leq |L_y|-1 + |S_y \cap L_y|$ must hold. Rearranging, we obtain $|S_y \cap L_y| \geq \alpha(|L_y| + 1) + 2$. Since $|S_y \cap L_y|$ is a non-negative integer, it must hold that $|S_y \cap L_y| \geq 2 + \lceil \alpha(|L_y|+1) \rceil$.

To prove sufficiency, we need to construct an allocation $S_y \cap L_y$ of $\lceil \alpha(|L_y|+1) \rceil + 2$ $y$ adversarial nodes such that $\phi(a;S) \geq \phi(a';S)$, where $a_i = y \ \forall i \in L_y$ and for any $a'$ deviating from $a$ in agents only in $L_y$. We first assume that $|L_y| \geq \lceil \alpha(|L_y|+1) \rceil + 2$, i.e. the length of the segment itself is greater or equal to the necessary number of adversaries needed. Indeed, let us define the sets $W_1$ and $W_2$ as follows:
\begin{align}
    W_1&=\{i \in L_y:\left\lfloor\alpha(i-j+1)\right\rfloor-\left\lfloor\alpha(i-j)\right\rfloor>0\}, \label{eq:W1} \\
    W_2&=\{j,w,j+|L_y|-1\}, \label{eq:W2}
\end{align}
where $w=\max\{i : i \in L_y\setminus(W_1 \cup \{j+|L_y|-1\}) \}$, i.e. the largest index that is neither in $W_1$ nor is the endpoint $j+|L_k|-1$. Then, set $S_y \cap L_y = W_1\cup W_2$.  An illustration of this influence set is depicted in Figure \ref{fig:placements}. Such a placement ``spreads out" adversaries along $L_y$ at a spacing of $\left\lceil \frac{1}{\alpha} \right\rceil$ nodes, and additionally places adversaries at the endpoints. We will show this placement ensures the sufficiency condition.

Let us denote $a_{L_y}$ as the actions of agents in $L_y$, $y_{L_y}$ as the partial profile where all agents in $L_y$ play $y$, and $x_{L_y}$ as when all agents in $L_y$ play $x$. Let us assume $S_x \cap L_y = \emptyset$. Any profile $a_{L_y} \notin \{y_{L_y},x_{L_y}\}$ belongs to one of three classes: 1) a single isolated $x$ segment $X_1$ within $L_y$, 2) an $x$ segment $X_2$ on the left and/or right edge of $L_y$, and 3) a combination of class 1) and 2).

For class 1 profiles, the potential of $y_{L_y}$ exceeds $a_{L_y}$ if 
\begin{equation}\label{eq:type1_condition}
	|S_y \cap X_1| > \alpha(|X_1|-1) - 2
\end{equation}
By construction of $S_y$, a lower bound on the number of $y$ adversaries influencing $X_1$ is $|S_y \cap X_1| \geq \left \lfloor \alpha|X_1|\right \rfloor$, which satisfies \eqref{eq:type1_condition}. Similarly for class 2 profiles, a sufficient condition is
\begin{equation}\label{eq:type2_condition}
	|S_y\cap X_2|>\alpha |X_2|.
\end{equation}
The number of adversaries on $X_2$ is given by $\left \lfloor \alpha|X_2|\right \rfloor +1$, which clearly satisfies \eqref{eq:type2_condition}. Alternatively, suppose $j+|L_y|-1 \in X_2$. The number of adversaries influencing $X_2$ is at least
\begin{equation}
    \begin{cases}
    |X_2| & \text{if } w\notin X_2,\\
    \lfloor\alpha(|X_2|-1)\rfloor+2 & else,
    \end{cases}
\end{equation}
since when $w \notin X_2$, $X_2 \subset S_y$. When $w\in X_2$, $\{w,j+|L_y|-1\} \in S_y$  in addition to the nodes in $W_1$. Both of these cases satisfy \eqref{eq:type2_condition}. Thus, $S_y$ satisfies the requirements of \eqref{eq:type1_condition} and \eqref{eq:type2_condition}. Consequently, the $S_y$ satisfies the requirement of type 3 profiles as well. By construction,  $|S_y|\leq\left \lfloor \alpha |L_k| \right \rfloor + 3$. If it is a strict inequality, one can simply add additional $y$ adversaries anywhere in the segment to meet the necessary condition \eqref{eq:y_nec}, i.e. the case when $a_{L_y} = x_{L_y}$. \hfill $\blacksquare$
% END PROOF

% Due to space limitations, the details of this assertion are contained in the included Supplementary materials. 
% \noindent\textbf{Proofs in Step 3}

%We can thus re-state the claim of property (3A): For any $a \in F(G,\gamma)$ such that $\boldsymbol{r}^{\top}(a)\boldsymbol{s}(a) > 0$, one can always find another graph $G' \in \mathcal{G}^1$ where there exists $a' \in F(G',\gamma)$ that satisfies $\eta(a') < \eta(a)$ and $\boldsymbol{r}^{\top}(a')\boldsymbol{s}(a') = 0$.
%
\begin{figure}
    \centering
	\includegraphics[scale=.6]{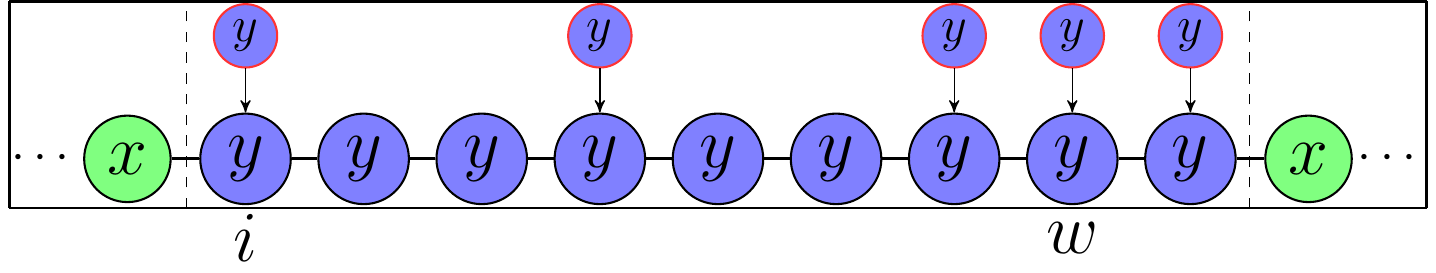}
	\caption{An illustration of the constructed influence set given by \eqref{eq:W1}, \eqref{eq:W2} to stabilize an isolated $y$ segment. The $y$ adversaries belonging to $S_y$ are depicted as the smaller circles attaching to agents (larger circles) in the network. In this example, $\alpha = \frac{1}{4}$ and $|L_y| = 9$. The necessary and sufficient number of adversaries to stabilize the segment is 5.}
	\label{fig:placements}
\end{figure}

\noindent\textbf{\underline{Proof of property (2A)}}

Suppose the minimum efficiency profile $a^*$ on the graph $G \in \mcal{G}^1$ is stabilized by the policy $\pi_S \in \Pi_{\text{I}}(G,\gamma)$, where all  adversaries are not utilized: $|S| < \lfloor \gamma\cdot n \rfloor$. The conditions \eqref{eq:y_nec} and \eqref{eq:x_nec} are met for all $x$ and $y$ segments, respectively. One can always add in remaining available $x$ adversaries ($y$) to the existing  $x$ ($y$) segments while retaining stability of $a^*$. Therefore, there exists a policy $\pi_S$ with $|S| = \lfloor \gamma\cdot n \rfloor$ that also stabilizes $a^*$. \hfill $\blacksquare$

\noindent\textbf{\underline{Proof of property (2B)}}

Before proving this property explicitly, we first define some relevant notations. We can describe a profile $a$ as a sequence of alternating segments $L_x^1L_y^1 L_x^2 L_y^2 \cdots$. For each unique segment pair pattern that appears in $a$, i.e. $|L_x|$ $x$ agents followed by $|L_y|$ $y$ agents, let us define the (column) vector $\boldsymbol{\ell}_x(a)$ whose elements  are the lengths $|L_x|$ among the unique patterns. We define $\boldsymbol{\ell}_y(a)$ similarly for the corresponding lengths $|L_y|$. Let us also define the vector $\boldsymbol{r}(a)$ whose elements are the number of times each unique pattern appears in $a$. We will drop the dependencies on $a$ when the context is clear. We refer to $\boldsymbol{r}$ as the \emph{repetition} vector. The efficiency of $a$ can be rewritten in the following suggestive form:
\begin{equation}\label{eq: gen eff}
    \eta(a) = \frac{{\boldsymbol{r}}^{\top} \left((1+\alpha)\boldsymbol{\ell}_x + \boldsymbol{\ell}_y\right) - (2+\alpha){||\boldsymbol{r}||}_1}{(1+\alpha){\boldsymbol{r}}^{\top}(\boldsymbol{\ell}_y + \boldsymbol{\ell}_x)}.
\end{equation}
We note that the denominator of \eqref{eq: gen eff} is simply the number of links in the ring, $n$, multiplied by $1+\alpha$. This indicates the optimal welfare $\frac{1}{2}W(a^{\text{opt}})$. The numerator of \eqref{eq: gen eff} counts (and weights with associated payoff) the number of coordinating $x$ and $y$ links given the description vectors $\boldsymbol{\ell}_x$, $\boldsymbol{\ell}_y$, and $\boldsymbol{r}$. For a profile $a$ and its associated description vectors $\boldsymbol{\ell}_x$ and $\boldsymbol{\ell}_y$, let us define the vector $\boldsymbol{s}(a)$ of identical length, whose components are given by $s_j(a) := \gamma({\ell}_{x,j} + {\ell}_{y,j})-\lceil \alpha({\ell}_{y,j} + 1)\rceil - 2	- \left\lceil (1+\alpha)^{-1}\left[ 2-\alpha({\ell}_{x,j} - 1) \right]_+ \right\rceil$. The number $s_j$ is the difference between the adversaries available to a particular segment pattern (given budget $\gamma$) whose length is given by ${\ell}_{x,j}$ and ${\ell}_{y,j}$, and the minimum number of adversaries needed to ensure its  stability (given by \eqref{eq:y_nec}, \eqref{eq:x_nec}). We refer to $\boldsymbol{s}$ as the \emph{surplus} vector. The quantity $\boldsymbol{r}^{\top}\boldsymbol{s}$ is the excess budget after using the minimum required number of adversaries to stabilize $a$.  Property (2A) asserts that a target profile of minimum efficiency satisfies $\boldsymbol{r}^{\top}\boldsymbol{s} = 0$. 

%
% \begin{equation}\label{eq:sm}
%     \begin{aligned}
%         s_j(a) := \gamma({\ell}_{x,j} &+ {\ell}_{y,j})-\lceil \alpha({\ell}_{y,j} + 1)\rceil - 2\\
% 	&- \left\lceil (1+\alpha)^{-1}\left[ 2-\alpha({\ell}_{x,j} - 1) \right]_+ \right\rceil
%     \end{aligned}
% \end{equation}
%

% The set of target  profiles \eqref{eq:feasible} can consequently be re-written as
% %
% \begin{equation}
%     F(G,\gamma) = \left\{a \in \mcal{A} : \boldsymbol{r}^{\top}(a)\boldsymbol{s}(a) \geq 0 \right\}.
% \end{equation}
% %
%%%%%%%%%%%%%%%%%%%%%%%%%%%%%%%%%%%%%%%%%%%%%%%%%%%%
Now, consider an action profile $a^1$ with $\boldsymbol{\ell}_x^1 = (\ell_{x,1},\ell_{x,2},\ell_{x,3})$, $\boldsymbol {\ell}_y^1 = (\ell_{y,1},\ell_{y,2},\ell_{y,3})$ and $\boldsymbol{s} = (s_1,s_2,s_3)$ with $s_1 > 0$ and  $s_2$, $s_3 < 0$. Hence, we can find $\boldsymbol{r}^1$ such that $(\boldsymbol{r}^1)^{\top} \boldsymbol{s}^1=0$. Thus, $a^1$ is a candidate for a minimum efficiency stable state. Furthermore, consider the profiles $a^2$ and $a^3$ (possibly defined on different ring graphs), where
$a^2$ is associated with $\boldsymbol{\ell}_x^2 = (\ell_{x,1},\ell_{x,2})$ and  $\boldsymbol{\ell}_y^2 = (\ell_{y,1},\ell_{y,2})$, and $a^3$ is associated with $\boldsymbol{\ell}_x^3 = (\ell_{x,1},\ell_{x,3})$ and  $\boldsymbol{\ell}_y^3 = (\ell_{y,1},\ell_{y,3})$. One can find repetition vectors $\boldsymbol{r}^2$, $\boldsymbol{r}^3$ that satisfy $(\boldsymbol{r}^2)^{\top} \boldsymbol{s}=0$ and $(\boldsymbol{r}^3)^{\top} \boldsymbol{s}=0$.

Define $g_i = \ell_{y,i} + (1+\alpha)\ell_{x,i} - (2+\alpha)$ and $\ell_i=\ell_{x,i}+\ell_{y,i}$ for each $i=1,2,3$. We can express efficiency of $a^1$ as $\eta(a^1)=\frac{r_2^1(g_2-\frac{s_2}{s_1}g_1) + r_3^1 (g_3-\frac{s_3}{s_1}g_1)}{(1+\alpha)(r_2^1(\ell_2-\frac{s_2}{s_1}\ell_1) + r_3^1 (\ell_3-\frac{s_3}{s_1}\ell_1))}$. One can write the efficiencies of $a^2$, $a^3$ as 
$\eta(a^2)=\frac{g_2-\frac{s_2}{s_1}g_1}{(1+\alpha)(\ell_2-\frac{s_2}{s_1}\ell_1)}$ and $\eta(a^3)=\frac{g_3-\frac{s_3}{s_1}g_1}{(1+\alpha)(\ell_3-\frac{s_3}{s_1}\ell_1)}$.
%
% \begin{equation}
%     \eta(a^1)=\frac{r_2^1(g_2-\frac{s_2}{s_1}g_1) + r_3^1 (g_3-\frac{s_3}{s_1}g_1)}{(1+\alpha)(r_2^1(\ell_2-\frac{s_2}{s_1}\ell_1) + r_3^1 (\ell_3-\frac{s_3}{s_1}\ell_1))}.
% \end{equation}
%
%
% \begin{equation}\label{eq:reduced_eta}
%     \eta(a^2)=\frac{g_2-\frac{s_2}{s_1}g_1}{(1+\alpha)(\ell_2-\frac{s_2}{s_1}\ell_1)},~ \eta(a^3)=\frac{g_3-\frac{s_3}{s_1}g_1}{(1+\alpha)(\ell_3-\frac{s_3}{s_1}\ell_1)}.
% \end{equation}
%
Observe that $\eta(a^1)$ is a mediant sum of weighted values $\eta(a^2)$ and $\eta(a^3)$. Hence, either $\eta(a^2)$ or $\eta(a^3)$ is less than or equal to $\eta(a^1)$. This result can be extended in a similar way to show that for any profile consisting of multiple segment patterns, one can construct another profile of lower efficiency using up to two unique segment patterns from the original action profile. \hfill $\blacksquare$ 
% END PROOF

\noindent\textbf{\underline{Proof of Step 3} (Theorem \ref{theorem:SI_opt})}

Using the collection of results we have obtained in Steps 1-3, we can now prove Theorem \ref{theorem:SI_opt}. From property (2B), the search for a minimal efficiency stable state, i.e., one that gives the efficiency \eqref{eq:SI_problem}, reduces to finding four lengths: $\boldsymbol{\ell}_x = (\ell_{x,1},\ell_{x,2})$ and $\boldsymbol{\ell}_y = (\ell_{y,1},\ell_{y,2})$. The form of the objective function in the integer program of \eqref{eq:SI_opt} thus coincides with the expression for $\eta(a^2)$ in property (2B).  Each $\ell_{z,i}$, $z\in{x,y}$ and $i\in\{1,2\}$,  must satisfy the length criterion $\ell_{y,i} \geq \left\lceil \frac{2+\alpha}{1-\alpha} \right\rceil$ and $\ell_{x,i} \geq 2$ ($3$ if $\alpha = 0$). These length conditions are consequences of the stabilizability conditions \eqref{eq:y_nec} and \eqref{eq:x_nec}. Lastly, one can find a repetition vector $\boldsymbol{r}$ that satisfies $\boldsymbol{r}^{\top} \boldsymbol{s}=0$, as long as $s_1 > 0$ and $s_2 < 0$, or $s_1 = 0$ with $\ell_{x,1}, \ell_{x,2} = 0$. \hfill $\blacksquare$

%
% An immediate consequence of this Lemma is the lengths of $x$ and $y$ segments must satisfy
% \begin{equation}\label{eq:minlengths}
% 		|L_y| \geq \max\left\{\frac{2+\alpha}{1-\alpha},3\right\}, \quad
% 		|L_x| \geq 2 .
% \end{equation}
%

%%%%%%%%%%%%%%%%%%%%%%%%%%%%%%%%%%%%%%%%%%%%%%%%%%%%%%%%%%%%%%%%%%%%%%%%

\subsection{Proof of Theorem \ref{theorem:SU_opt}}

Here, we provide a proof of Theorem \ref{theorem:SU_opt}, which characterizes the minimal efficiency a static uninformed adversary can induce on a $k$-connected ring graph. The arguments rely on an extension of intermediate step 1 from the proof of Theorem \ref{theorem:SI_opt} to $k$-connected ring graphs. 

In particular, the necessary and sufficient condition to stabilize a $y$ segment in a $k$-connected ring graph is \begin{equation}\label{eq:ystable_kring}
		|S_y \cap L_y| \geq \left\lceil \alpha\left( k |L_y| + \frac{k(k+1)}{2} \right) \right\rceil + k(k+1), 
\end{equation}
and the spacing between two sequential $y$ adversaries within $L_y$ is no more than $\left\lceil \frac{1}{k\alpha} \right\rceil$. Note that according to this condition, the segment length must also satisfy $|L_y| \geq \max\left\{1,\lceil \frac{ k(k+1)(1 + \alpha/2) }{1-k\alpha} \rceil \right\}$. A derivation of the condition is as follows. There are $\sum_{j=1}^k (|L_y| - j)$ links between  agents in $L_y$. Assuming $|L_y|$ satisfies the length requirement, there are $2\sum_{j=1}^k j$ links from $L_y$ to outside $L_y$. The potential of $y_{L_y}$ (all agents in $L_y$ play $y$) exceeds that of $x_{L_y}$ (all play $x$) if $|S_y \cap L_y| + \sum_{j=1}^k (|L_y| - j) \geq (1+\alpha)\left[ \sum_{j=1}^k (|L_y| - j) + 2\sum_{j=1}^k j \right]$, which reduces to \eqref{eq:ystable_kring}. One can prove sufficiency in a similar manner as step 1 from the previous section -- by allocating the $y$ adversaries with a spacing of $\left\lceil \frac{1}{k\alpha} \right\rceil$ apart, the potential of $y_{L_y}$ exceeds that of any other $a_{L_y} \neq \{y_{L_y},x_{L_y} \}$.

We are now ready to prove Theorem \ref{theorem:SU_opt}.
    A static and uninformed policy cannot strategically place adversarial nodes in the network. It can only specify the numbers of $x$ and $y$ adversaries. Its baseline performance is given by the minimal damage that can be inflicted over all possible allocations of these adversaries. Hence to characterize \eqref{eq:SU_OPT},  we seek the allocation of adversaries that ensures the best-case efficiency for the network.

    First, we consider the case $\gamma < k\alpha$. The adversarial nodes can be allocated sparsely enough across the entire network such that the condition \eqref{eq:ystable_kring} is violated. Consequently, the all $x$ profile is the unique stochastically stable state. Therefore, no damage can be inflicted on the system in this regime. Note that if $k\alpha > 1$, no damage is possible regardless of the budget.
    
    Now, consider $\gamma > k\alpha$. If $k\alpha < 1$ and $G \in \mcal{G}^k$ is sufficiently large, an allocation of $y$ adversaries according to \eqref{eq:ystable_kring} would ensure conversion of the entire network to $y$, giving an efficiency of $\frac{1}{1+\alpha}$. However, let us consider a  re-allocation of these adversaries that maximally mitigates such damage. The idea is to only allow a minimal fraction $f$ of the network to be converted to $y$, while the rest of the network plays $x$. 
    
    Suppose $y$ adversaries are allocated to every agent in a contiguous segment, whose length is a fraction $f$ of the entire network. Suppose this segment is sufficiently long such that \eqref{eq:ystable_kring} is satisfied. Now, the remaining $\gamma - f$ adversaries  should be allocated to the rest of the network such that the remaining fraction $1-f$ of the network (another contiguous segment) is still stable to $x$. Indeed, an adversarial agent density of up to $k\alpha$ in the remaining network fails to induce any $y$ agents. The smallest $f$ that satisfies these conditions is given by $f = \frac{\gamma - k\alpha}{1 - k\alpha}$. This establishes \eqref{eq:SU_OPT}. Note in this analysis, the adversary exclusively chooses to implement $y$ adversaries. Based on the above arguments, an optimal static uninformed policy never chooses to use $x$ adversaries.

%%%%%%%%%%%%%%%%%%%%%%%%%%%%%%%%%%%%%%%%%%%%%%%%%%%%%%%
\section{Proofs: Performance of dynamic policies}\label{sec:dynamic_analysis}
% !TEX root = SM.tex
% !TEX root = SM.tex
In this section, we give proofs for the minimum efficiency dynamic adversaries can induce. Similar to Section \ref{sec:static_analysis}, we will first prove Theorem \ref{theorem:MI_opt}, the case of a dynamic informed adversary.  We then give the proof of Theorem \ref{theorem:DU_opt}.

% An outline of this proof follows the same basic three-step structure as that of Theorem \ref{theorem:SI_opt}. We will hence follow this general structure for Theorem \ref{theorem:MI_opt}.

Due to the time-dependent nature of dynamic policies, we cannot rely on potential game arguments to compute stochastically stable states as we did in Section \ref{sec:static_analysis}. One must instead leverage the theory of regularly perturbed Markov processes and resistance trees. Before delving into the proof of Theorem \ref{theorem:MI_opt}, we provide a brief overview of this theory below. More detailed treatments can be found in \cite{Young1993,Young_2001}.

\subsection{Preliminary: Regularly perturbed Markov processes and resistance trees}

\begin{definition}
	A Markov process with transition matrix $P^\epsilon$ defined over state space $\mcal{A}$ and parameterized by a perturbation $\epsilon \in (0,\bar{\epsilon}]$ for some $\bar{\epsilon}>0$ is a \emph{regular perturbation of the process} $P^0$ if it satisfies:
	\begin{enumerate}
		\item $P^\epsilon$ is aperiodic and irreducible for all $\epsilon \in (0,\bar{\epsilon}]$.
		\item $\lim_{\epsilon \rightarrow 0^+} P^\epsilon(a,a') \rightarrow P^0(a,a')$ for all $a,a' \in \mcal{A}$.
		\item If $P^\epsilon(a,a') > 0$ for some $\epsilon \in (0,\bar{\epsilon}]$ then there exists  $r(a,a') \geq 0$ such that $0 < \lim_{\epsilon \rightarrow 0^+} \frac{P^\epsilon(a,a')}{\epsilon^{r(a,a')}} < \infty$.
		We call $r(a,a')$ the \emph{resistance} of transition $a \rightarrow a'$.
	\end{enumerate}
\end{definition}
The log-linear learning process is a regularly perturbed process with error parameter $\epsilon = e^{-\beta}$. The transition graph of $P^\epsilon$ is a directed graph whose nodes are the action profiles $\mcal{A}$ and the edge $(a,a')$ exists if and only if $P^\epsilon(a,a') > 0$. The weights of such edges are given by the resistances $r(a,a')$. The resistance of a path of length $m$, $\zeta = (z_1\rightarrow z_2 \rightarrow \dots \rightarrow z_m)$, is the sum of resistances along the state transitions: $r(\zeta) := \sum_{k=1}^{m-1} r(z_k,z_{k+1})$. Let us denote the \emph{recurrent classes} of the unperturbed process $P^0$ as $E_1,E_2,\dots,E_N$ with $N\geq 1$ where each class $E_k \subset \mcal{A}$. A recurrent class satisfies the following.
\begin{enumerate}
	\item For all $a \in \mcal{A}$, there is a zero resistance path from $a$ to $E_k$ for some $k\in\{1,\ldots,N\}$.
	\item For all $k \in \{1,\ldots,N\}$, and all $a,a' \in E_k$, there exists a zero resistance path from $a$ to $a'$ and from $a'$ to $a$.
	\item For all $a,a'$ with $a \in E_k$ for some $k \in \{1,\ldots,N\}$ and $a' \notin E_k$, $r(a,a') > 0$.
\end{enumerate}
One can also consider another directed transition graph whose nodes are the $N$ recurrent classes. In this graph, all edges exist. Edge $(E_i,E_j)$ is weighted by $\rho_{ij}$, defined as the minimum resistance among paths in the action profile transition graph starting from $E_i$ and ending in $E_j$: $\rho_{ij} := \min_{a\in E_i, a' \in E_j} \min_{\zeta \in \mcal{P}(a\rightarrow a')} r(\zeta) $. where $\mcal{P}(a\rightarrow a')$ denotes the set of all paths starting at $a$ and ending at $a'$. Let $\mcal{T}_k$ be the set of all spanning trees rooted in the class $E_k$. That is, an element of $T \in \mcal{T}_k$ is a directed graph with $N-1$ edges such that there is a unique path from $E_j$ to $E_k$, for every $j \neq k$. The resistance $R(T)$ of the rooted tree $T$ is the sum of resistances $\rho_{ij}$ on the $N-1$ edges that compose it. Now, define $\psi_k := \min_{T \in \mcal{T}_k} R(T)$ as the \emph{stochastic potential} of recurrent class $E_k$. We will use the following result to identify stochastically stable states.

% \begin{equation}
% 	\rho_{ij} := \min_{a\in E_i, a' \in E_j} \min_{\zeta \in \mcal{P}(a\rightarrow a')} r(\zeta) 
% \end{equation}
\begin{lemma}[from \cite{Young1993}]\label{lem:stoch_potential}
	The state $a \in \mcal{A}$ is stochastically stable if and only if $a \in E_k$, where $k \in \argmin{j\in\{1,\ldots,N\}} \psi_j$. 	That is, it belongs to a recurrent class with minimum stochastic potential. It is the unique stochastically stable state if and only if $E_k = \{a\}$ and $\psi_k < \psi_j$, $\forall j \neq k$.
\end{lemma}

\subsection{Proof of Theorem \ref{theorem:MI_opt}}

The logic of the proof follows the same three-step structure as the proof of Theorem \ref{theorem:SI_opt}. The only component that differs are the necessary and sufficient budget conditions to stabilize $x$ and $y$ segments. Indeed, we expect a dynamic informed adversary to need fewer adversaries than its static counterpart.

An outline of the proof is as follows. We show a particularly defined dynamic policy, which we term an \emph{aggressive policy}, is sufficient to stabilize a given target profile. This entails proving that $a$ is the recurrent class of minimum stochastic potential (see Lemma \ref{lem:stoch_potential}). To do so, we characterize the set of recurrent classes and demonstrate that minimum resistance paths leaving each class leads to another that is more ``similar" to the target profile $a$. We then prove necessity -- any other dynamic policy utilizing strictly fewer adversarial nodes than the aggressive policy cannot stabilize $a$. We can then formulate an integer optimization problem similar to Theorem \ref{theorem:SI_opt}, but with different constraints on the number of adversaries needed for each $x$ and $y$ segment. To begin, we formally define the aggressive policy based on profile $a \in \mcal{A}$.

%We first identify minimum length requirements for $x$ and $y$ segments in target profiles that an aggressive policy can stabilize. We then characterize the set of recurrent classes, and demonstrate that the minimum resistance path leaving each class leads to another that is more ``similar" to the target profile $a$. This allows us to construct a tree rooted in $a$ that possesses minimum stochastic potential. We then prove necessity -- any other dynamic policy utilizing strictly fewer adversarial nodes than the aggressive policy cannot stabilize $a$. We can then formulate an integer optimization problem similar to Theorem \ref{theorem:SI_opt}, but with different constraints on the number of adversaries needed for each $x$ and $y$ segment. We omit proofs of some intermediate steps. They can be found in the supplementary material or online version \cite{Canty_2019}. To begin, we formally define the aggressive policy based on profile $a \in \mcal{A}$.

%
\begin{figure*}[t]
	\centering
	\includegraphics[scale=.4]{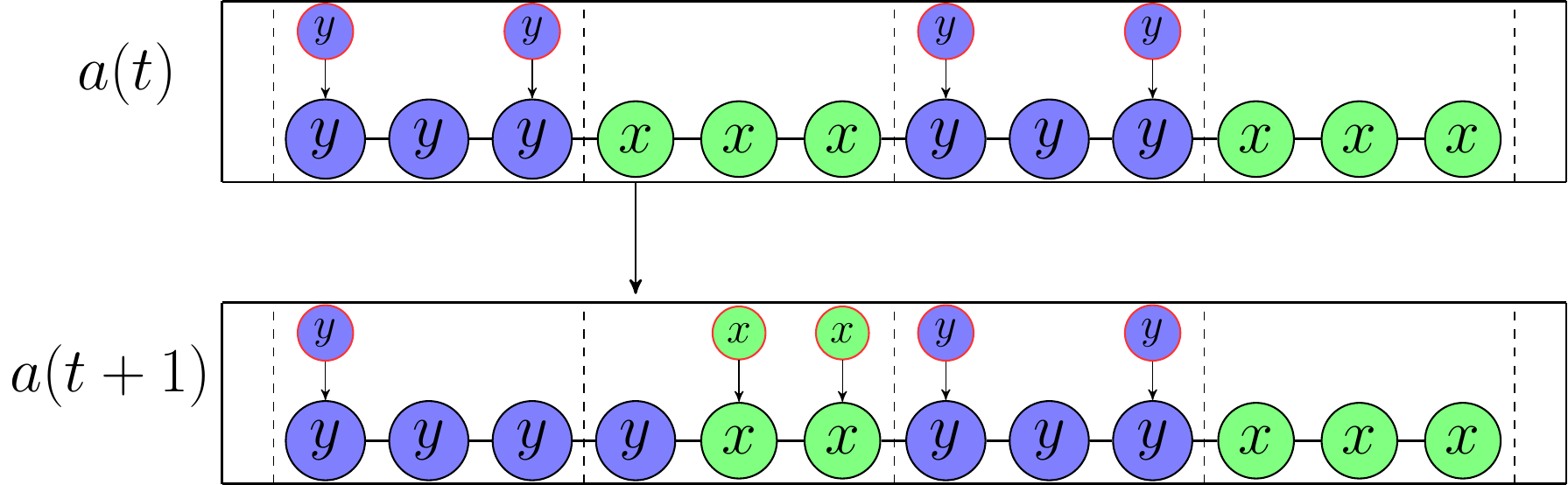} \quad
	\includegraphics[scale=.4]{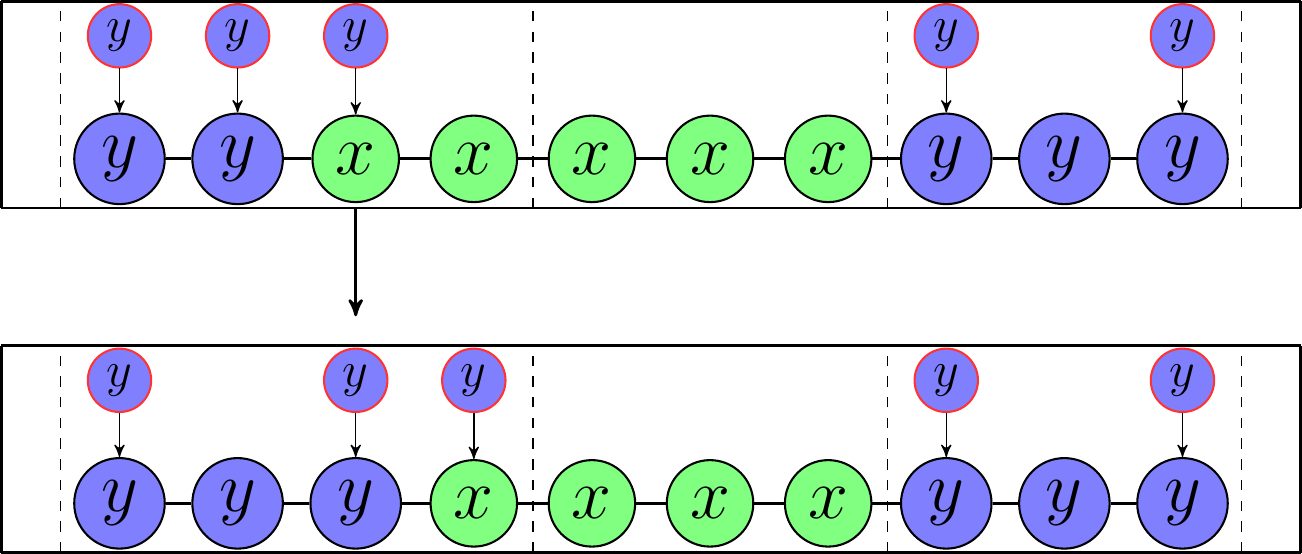}
	\caption{\small  We illustrate here both defensive and offensive strategies in an aggressive policy. (Left) Defensive $y$ strategies are applied the first and third segments from the left. The fourth agent from the left transitioning from $x$ to $y$ at time $t+1$ activates a defensive $x$ strategy in the second. No adversaries are deployed to $x$ segments until only two neighboring agents playing $x$ remain. (Right) A defensive and offensive $y$ strategy are applied simultaneously to the first segment. The offensive strategy attaches a $y$ adversary to the $x$ agent that has a $y$ neighbor.   }
	\label{fig:aggressive}
\end{figure*}
\begin{definition}\label{def: agg policy}
  (Aggressive policy targeting $a$). An \emph{aggressive policy targeting} $a \in \mcal{A}$ is a state-dependent policy with adversarial placements $\{S_x(a(t))$, $S_y(a(t))\}_{t\geq 0}$ satisfying the following properties. 
%   \begin{enumerate}

	\noindent\textbf{1)} \emph{(Defensive $y$ strategy)} For each $y$-segment $L_y$ contained in $a$, suppose $[p,q] = \{p,p+1,\ldots,q\} \subseteq L_y$, with $p\neq q$, is the longest segment of agents within $L_y$ playing $y$ in $a(t)$. Then,
	    \begin{equation}
          	S_y(a(t)) \cap [p,q] =
        	\begin{cases}
        		\{p,q\} & \text{if } a_{p-1}(t)=a_{q+1}(t)=x \\
        		p & \text{if } a_{p-1}(t)=x, \ a_{q+1}(t)=y \\ 
        		q & \text{if } a_{p-1}(t)=y, \ a_{q+1}(t)=x 
    	    \end{cases}
        \end{equation}
    If the length of $[p,q]$ is one $(p=q)$, then $p \notin S_y(a(t))$. 
    
    \noindent\textbf{2)} \emph{(Defensive $x$ strategy)} For each $x$-segment $L_x$ contained in $a$, suppose $[p,q]$, $p\neq q$, is the longest segment of agents within $L_x$ playing $x$ in $a(t)$.
	    \begin{enumerate}[label=(\alph*)]
	        \item   If $\alpha < \frac{1}{2}$ and $|L_x| \leq 1 + \left\lfloor \frac{1-\alpha}{\alpha} \right\rfloor$, then
	        \begin{equation}
            	S_x(t) \cap [p,q] = 
            	\begin{cases}
            		\{p,q\} & \text{if } q-p = 1 \\
            		\varnothing & \text{otherwise}
            	\end{cases}
            \end{equation}
            If $|L_x| \geq 2 + \left\lfloor \frac{1-\alpha}{\alpha} \right\rfloor$, then $S_x(t) \cap L_x = \varnothing$.
            \item Suppose $\alpha \geq \frac{1}{2}$. Then $S_x(t) \cap L_x = \varnothing$.
	    \end{enumerate}

	\noindent\textbf{3)} \emph{(Offensive strategies)} Consider the segment of lowest index that is not aligned, i.e. $a_i(t) \neq y$ for at least one $i \in L_y = [u,v]$, for a $y$-segment. Then the following properties hold for $S_y(t)$. A similar implementation holds if it is an $x$-segment.
        \begin{enumerate}[label=(\alph*)]
            \item  Denote $[p,q] \subset L_y$ as the longest segment of agents within $L_y$ playing $y$ in $a(t)$. Then $S_y(t) \cap L_y$ contains either $p-1$ or $q+1$, but not both.
            \item If $a_i(t) = x$ for all $i \in L_y$, then $S_y(t) \cap L_y$ contains
            \begin{itemize}
                \item $u$ or $v$ (but not both), when $a_{u-1}(t) = a_{v+1}(t)$.
                \item $u$, when $a_{u-1}(t) = y$ and $a_{v+1}(t) = x$.
                \item $v$, when $a_{u-1}(t) = x$ and $a_{v+1}(t) = y$.
            \end{itemize}
        \end{enumerate}
%   \end{enumerate}
\end{definition}

Properties 1 and 2 describe ``defensive $y$ (resp. $x$)" strategies to maintain $y$ ($x$) segments over time. A defensive $y$ strategy is implemented on all $y$-segments at any given time. In property 2, defensive $x$ strategies are applied only if $\alpha < \frac{1}{2}$, and to segments that are shorter than a threshold length. Furthermore, the strategy does not ``activate" until there are at most two consecutive agents playing $x$ in the segment. Property 3 describes ``offensive" strategies that are intended to convert segments back to their original type $x$ or $y$. Note that the aggressive policy applies an offensive strategy to only a single segment at any time $t$, if needed.  Figure \ref{fig:aggressive} depicts an illustration of allocations of adversarial nodes to segments according to an aggressive policy. 

Let $n_y$ be the number of $y$ segments and $n_x$ the number of $x$ segments of length at most $1 + \left\lfloor \frac{1-\alpha}{\alpha} \right\rfloor$ in profile $a$. Then, the minimum number of adversarial nodes needed to implement an aggressive policy targeting $a$ is
\begin{equation}
	\begin{cases}
		2(n_y + n_x) + 1, \text{ if } \alpha < \frac{1}{2} \\
		2n_y + 1, \text{ if } \alpha \geq \frac{1}{2} \\
	\end{cases}
\end{equation}

Here, the additional $+1$ adversary is needed to implement an offensive strategy. We now establish some basic facts and terminologies. Under log-linear learning, states transition via unilateral agent deviations. If two profiles $a^1$ and $a^2$ differ by agent $i$'s deviation, the resistance is
\begin{equation}
	r(a^1 \rightarrow a^2) = \left[ \tilde{U}_i(a^1_i,a^1_{-i};S(a^1)) - \tilde{U}_i(a^2_i,a^1_{-i};S(a^1)) \right]_+
\end{equation}
where recall $\tilde{U}_i$ is agent $i$'s perceived utility \eqref{eq:perceived_u}. To include more specificity regarding such deviations, we introduce the following notation. Suppose in profile $a^1$,  agent $i$ plays $z \in \{x,y\}$, has $b\in \{0,1,2\}$ neighbors also playing $z$, and is influenced by an adversarial node of type $s \in \{x,y,\varnothing\}$. In $a^2$, agent $i$ unilaterally deviates to $\{x,y\}\setminus z$. Then we write
\begin{equation}
	\begin{aligned}
		\omega(z,b,s) &:= r(a^1 \rightarrow a^2)
	\end{aligned}
\end{equation}
to denote the magnitude of the resistance for this transition. The following result characterizes the minimum required lengths for $x$ and $y$ segments in a target profile.
\begin{lemma}\label{lem:a_lengths}
	If $a\in\mcal{A}$ is stochastically stable under the aggressive policy targeting $a$, then
	\begin{itemize}
		\item all $x$ segments of $a$ are of length 2 or greater. 
		\item all $y$ segments of $a$ are of length $\lceil \frac{2+\alpha}{1-\alpha} \rceil$ or greater. 
	\end{itemize} 
\end{lemma}
\begin{proof} 
Suppose $(a_{i-1}, a_i, a_{i+1}) = (y,x,y)$ for some agent $i$. Regardless of what adversarial policy is applied, there is a zero resistance path out of $a$. Specifically, $r(a \rightarrow a') = 0$ where $a_i = y$ and $a'_{-i} = a_{-i}$. The least resistant path from $a'$ to $a$ is  $1-\alpha > 0$, possible if and only if $i \in S_x(a')$. Hence, $a$ is not a recurrent class and therefore is not stochastically stable.

		Suppose $a$ has a $y$ segment $L_y$ and $|L_y| \leq \left\lceil \frac{1+2\alpha}{1-\alpha} \right\rceil$. Let $a'$ be the similar profile with $a'_{L_y} = x_{L_y}$. Note that $a'$ is a recurrent class. When the aggressive policy applies an offensive strategy on $L_y$, the minimum resistance path starting from $a'$ and ending in $a$ is given by a border agent's $x\rightarrow y$ transition (having resistance $1+2\alpha$), followed by each subsequent neighbor's $x\rightarrow y$ transition (each having resistance $0$). The resistance of this path is $\rho_{a',a} = 1 + 2\alpha$.

		The minimum resistance path starting from $a$ and ending in $a'$ consists of $|L_y|-1$ transitions of resistance $1-\alpha$. Hence, $\rho_{a, a'} = (1-\alpha)(|L_y|-1) < (1-\alpha) \frac{1+2\alpha}{1-\alpha}  = 1 + 2\alpha$.
% 		\begin{equation}\label{eq:res1}
% 			\begin{aligned}
% 				\rho_{a, a'} &= (1-\alpha)(|L_y|-1) \leq (1-\alpha)\left( \left\lceil \frac{1+2\alpha}{1-\alpha} \right\rceil - 1 \right) \\
% 				&< (1-\alpha) \frac{1+2\alpha}{1-\alpha}  = 1 + 2\alpha.
% 			\end{aligned}
% 		\end{equation}
		Let $T$ be the minimum resistance tree rooted in $a$, and note that the edge $(a',a)$ is necessarily part of $T$. Consider the tree $T'$ rooted in $a'$ by replacing the edge $(a',a)$ from $T$ with $(a,a')$. $T'$ has  lower stochastic potential than $T$ and therefore $a$ is not stochastically stable.
\end{proof}
Henceforth, we only consider target profiles $a$ with properties given by Lemma \ref{lem:a_lengths}. In the forthcoming analysis, we characterize the set of recurrent classes induced by the aggressive policy. We first define terminology to describe any profile $a'$ relative to $a \in \mcal{A}$. Let $L_y$ be a $y$-segment contained in $a$. We say the segment $L_y$ is \emph{homogeneous in $a'$} if $a'_i = a'_j$ for all $i,j \in L_y$. We say $L_y$ is \emph{heterogeneous} if it is not homogeneous. Similar terminology applies for $x$-segments of $a$. We say the profile $a'$ is homogeneous if every $x$ and $y$-segment is homogeneous in $a'$, and it is heterogeneous if it is not homogeneous. We will denote particular portions of a homogeneous action profile $a'$ with brackets $|_x$ and $|_y$ that separate the segments based on the target profile $a$. For instance, $|X|_x X|_y$ refers to the actions of agents in an action profile for two consecutive segments with all agents playing $x$ in the first as well as the second. The subscripts convey that agents play $x$ in the target profile $a$ in the first segment and $y$ in the second segment. We will often compare two homogeneous action profiles that differ only in one segment, and term the two profiles \emph{similar}.

The following result characterizes the set of all recurrent classes induced by the aggressive policy. In particular, each recurrent class consists of a single homogeneous action profile.
\begin{lemma}\label{lem:AP_rec_class}
	The recurrent classes associated with the aggressive policy targeting $a$ are the homogeneous action profiles that do not contain an instance of $|X|_y Y|_x$.

% that can be constructed using only instances of $\{|X|_x Y|_y X|_x\}$, $\{|Y|_y X|_x Y|_y\}$, $\{|X|_x X|_y X|_x\}$, and $\{|Y|_y Y|_x Y|_y\}$. 
\end{lemma}

\begin{proof}
	We can disqualify any action profile $a^1$ having at least one heterogeneous segment $L$ from being a recurrent class. The reason is that an offensive strategy induces a zero-resistance path from $a^1$ to $a^2$, where $a^2_{L}$ is homogeneous. Furthermore, the resistance of any path from $a^2$ back to $a^1$ is necessarily non-zero.

    Thus, recurrent classes must be homogeneous profiles. Observe any homogeneous profile containing an instance of $|X|_y Y|_x$ has a zero-resistance path to a profile where that instance is replaced by $|X|_y X|_x$. This is because an aggressive policy applies an offensive $x$ strategy on the $Y|_x$ segment. Any homogeneous profile that does not contain an instance of $|X|_y X|_x$ is necessarily composed of instances of $|X|_x Y|_y X|_x$, $|Y|_y X|_x Y|_y$, $|X|_x X|_y X|_x$, and $|Y|_y Y|_x Y|_y$. Under an aggressive policy, there are no zero-resistance paths out of any of these segment patterns.
\end{proof}
Let us denote $\mcal{A}_{\text{R}} \subset \mcal{A}$ the set of recurrent classes. We can thus focus our attention to homogeneous action profiles described in Lemma \ref{lem:AP_rec_class} as candidates for stochastically stable states, which includes the target profile $a$ itself. For the next steps in the proof, we need the following calculations regarding minimum resistances between similar recurrent classes.

\begin{lemma}\label{lem:resistances}
	Consider the following transitions between two similar recurrent classes $a^1, a^2 \in\mcal{A}_{\text{R}}$. If the transition is
	\begin{enumerate}
		\item  $|Y|_y Y|_x Y|_y \rightarrow |Y|_y X|_x Y|_y$, then $\rho_{a^1,a^2} = 1-\alpha$.
		\item  $|X|_x X|_y X|_x \rightarrow |X|_x Y|_y X|_x$, then $\rho_{a^1,a^2} = 1+2\alpha$.
		\item $|X|_x Y|_y X|_x \rightarrow |X|_x X|_y X|_x$, then $\rho_{a^1,a^2} > 1+2\alpha$.
		\item $|X|_x \rightarrow|Y|_y$, then $\rho_{a^1,a^2} > 1-\alpha$.
	\end{enumerate}
\end{lemma}
\begin{proof}
		\noindent 1) When an aggressive policy implements an offensive $x$ strategy on the middle segment $L_x$, a path from $a^1$ to $a^2$ consists of $|L_x|$ unilateral switches from $y$ to $x$. The first switch requires $\omega(y,2,x) = 1-\alpha$ resistance. The rest switch with resistance $\omega(y,1,x) = 0$. The resistance of any other path necessarily is greater than $1-\alpha$
		% (e.g. when the middle segment is not selected for an offensive strategy).

		\noindent 2)  The least resistant path occurs when an offensive $y$ strategy is applied to the middle segment $L_y$, and the agents along the segment sequentially switch. This path requires one deviation of resistance $\omega(x,2,y) = 1+2\alpha$, and the other $|L_y|-1$ deviations of resistance $\omega(x,1,y) = 0$.

		\noindent 3)  A defensive $y$ strategy is applied to the middle segment $L_y$. The path $\zeta$ in which each agent sequentially deviates requires $|L_y|-1$ deviations of type $\omega(y,1,y) = 1-\alpha$.  By Lemma \ref{lem:a_lengths}, $|L_y| > \lceil \frac{1+2\alpha}{1-\alpha} \rceil \Rightarrow r(\zeta) > 1 + 2\alpha$. Another path requires at least one deviation of an agent in the middle of the segment, which has resistance $\omega(y,2,\varnothing) = 2$. Hence, $\rho_{a^1,a^2} > 1+2\alpha$.

		\noindent 4) A transition of this form either
		
		\begin{itemize}
			\item has at least one deviation of type $\omega(x,2,\varnothing) = 2(1+\alpha) > 1-\alpha$, e.g. an agent with two $y$ neighbors switches.
				
			\item has at least one deviation of type $\omega(x,1,x) = 1+2\alpha> 1-\alpha$, e.g. a defensive $x$ strategy is applied.
				
			\item only has deviations of type $\omega(x,1,\varnothing) = \alpha$, i.e. no defensive $x$ strategy applied. This is the case if $\alpha \geq \frac{1}{2}$, in which $\alpha > 1-\alpha$. The other case is if $\alpha \leq \frac{1}{2}$ and $|L_x| \geq 2 + \left\lfloor \frac{1-\alpha}{\alpha} \right\rfloor$. It takes $|L_x| - 1$ transitions of type $\omega(x,1,\varnothing) = \alpha$, for which $\alpha(|L_x|-1) > 1-\alpha$. 
		\end{itemize}
\end{proof}

Every $a' \in \mcal{A}_{\text{R}}$ can be assigned a level of ``disagreement" $d(a')$ corresponding to the number of homogeneous segments that differ relative to their counterparts in target profile $a$:
\begin{equation}
	d(a') := |\{L_z | a'_{L_z} \neq a_{L_z}, z \in \{x,y\} \}|.
\end{equation}
The next result demonstrates that disagreement decreases along minimum resistance paths between recurrent classes. 

\begin{lemma}\label{lem:subgraphs}
	Consider the directed graph $\Sigma= (\mcal{A}_{\text{R}},\mcal{E})$ in which the edges $\mcal{E}$ are formed by connecting recurrent classes through the minimum resistance edge leaving each class.  Then $\Sigma$ is composed of a  collection of disconnected subgraphs $\Sigma_u = (\mcal{A}_u,\mcal{E}_u)$, each one corresponding to a particular recurrent class $u$. Each subgraph $\Sigma_u$ has the following properties.
	\begin{itemize}
		\item The class $u$ belongs to $\Sigma_u$, and for every node $v \in \mcal{A}_u$, $v \neq u$, there is a unique path from $v$ to $u$. 
		\item There exists a class $v \in \mcal{A}_u$ s.t. $(u,v), (v,u) \in \mcal{E}_u$.
		\item $u = \argmin{v \in \mcal{A}_u} d(v)$.
	\end{itemize}
	We refer to the class $u$ as the \emph{head} of the subgraph $\Sigma_u$.
\end{lemma}
\begin{proof}
By Lemma \ref{lem:AP_rec_class}, each recurrent class is a homogeneous action profile not containing an instance of $|X|_y Y|_x$. Let \begin{equation}
	\mcal{A}_{\text{R}}^0 := \{ a' \in \mcal{A}_{\text{R}} : a' \text{ has an instance of $|Y|_y Y|_x Y|_y$ } \}.
\end{equation}
Suppose $a^1\in \mcal{A}_{\text{R}}^0$. By Lemma \ref{lem:resistances}, the edge $(a^1,a^2) \in \mcal{E}$ has resistance $1-\alpha$, since $a^2$ is similar to $a^1$ where an instance of $|Y|_y Y|_x Y|_y$ in $a^1$ is replaced with $|Y|_y X|_x Y|_y$. Thus, any $(a^1,a^2) \in \mcal{E}$ with $a^1 \in \mcal{A}_{\text{R}}^0$ satisfies $d(a^2) = d(a^1) - 1$. We term this type of edge a ``type 0" transition.

Now, consider any class $a^1 \in \mcal{A}_{\text{R}} \setminus \mcal{A}_{\text{R}}^0$. The minimum resistance edge leaving $a^1$ to some $a^2 \in \mcal{A}_{\text{R}}$ is one of the following types (with resistances due to Lemma \ref{lem:resistances}).
\begin{enumerate}
	\item An instance $|X|_x X|_y X|_x$ replaced by $|X|_x Y|_y X|_x$, with resistance $r_1(a^1) = 1+2\alpha$. Then $d(a^2) = d(a^1)-1$.
	\item An instance  $|X|_y X|_x Y|_y$  replaced by $|Y|_y Y|_x Y|_y$, with resistance $r_2(a^1) > 1-\alpha$. Then $d(a^2) = d(a^1)$.
	\item An instance $|X|_x Y|_y X|_x$ replaced by $|X|_x X|_y X|_x$, with resistance $r_3(a^1) > 1+2\alpha$. Then $d(a^2) = d(a^1)+1$.
	\item An instance $|Y|_y X|_x Y|_y$ replaced by $|Y|_y Y|_x Y|_y$, with resistance $r_4(a^1) > 1-\alpha$. Then $d(a^2) = d(a^1)+1$.
\end{enumerate}
For $a^1 = a$, the minimum resistance path can only be of type 3 or 4.   Thus, the minimum resistance path out of $a$ must increase disagreement. Subsequently, $(a,a^2) \in \mcal{E}$ and $(a^2,a) \in \mcal{E}$. To see this, the case $(a,a^2)$ being  type 4 follows from previous arguments. For the case $(a,a^2)$ of type 3, it necessarily holds that $r_3(a^2), r_4(a^2) > 1+2\alpha$. Hence, the edge $(a^2,a)$ exists and is of type 1. 

For classes $a^1 \neq a$, it cannot be of type 3. If it is of type 1 or 2, the edge leads either to a class with lower disagreement or a class in $\mcal{A}_{\text{R}}^0$, which in turn leads to a class with lower disagreement. If $a^1 \neq a$ has a type 4 edge, $r_4(a^1) < 1+2\alpha$ and the path leads to a class  $a^2 \in \mcal{A}_{\text{R}}^0$ (with higher disagreement). Subsequently, the minimum resistance edge from $a^2$ leads back to $a^1$ with a type 0 edge. 

We can thus deduce for each class $u \neq a$ that has a type 4 edge, there is a corresponding subgraph $\Sigma_u   \subset \Sigma$ that has the following properties. All edges in $\mcal{E}_u$ except the edge leaving $u$ lead to a class that either has lower disagreement or has an edge connecting to another class with lower disagreement. Consequently, $u$ has the lowest disagreement of all classes in $\Sigma_u$, i.e. $u = \argmin{a' \in \mcal{A}_u} d(a')$.
\end{proof}

An illustration of a subgraph $\Sigma_u$ is shown in Figure \ref{fig:subgraph} (left). Note that the rooted tree on a subgraph $\Sigma_u$ with minimal stochastic potential is given by $\Sigma_u$ without the edge $(u,v)$. The next result asserts that the minimal resistance edge that leaves a subgraph leads to another subgraph whose head node has lower disagreement.
\begin{figure}
	\centering
	\includegraphics[scale=.42]{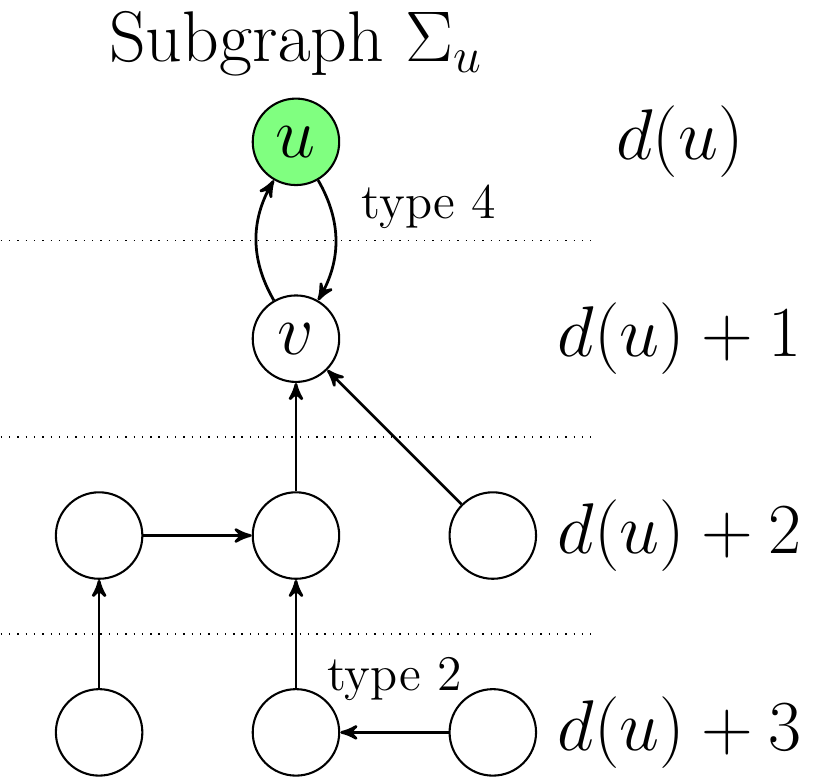}
	\includegraphics[scale=.15]{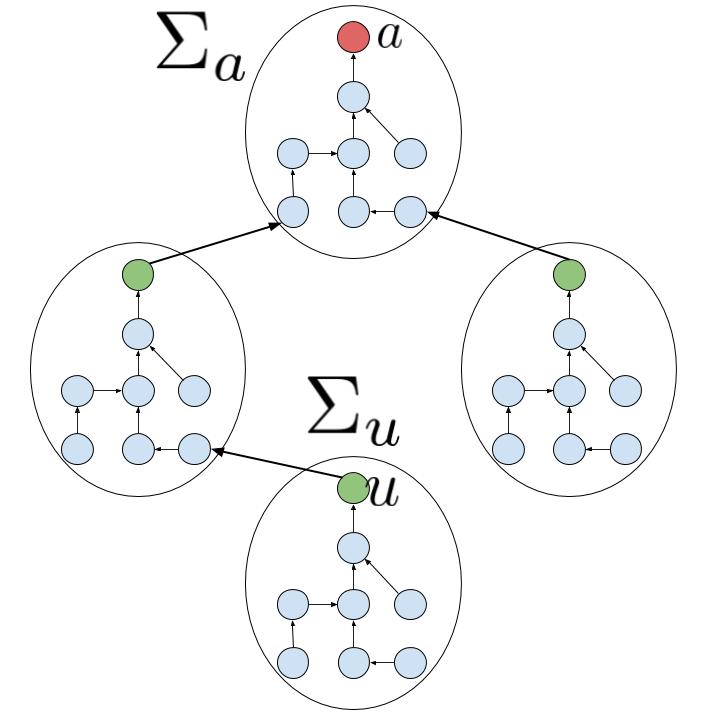}
	\caption{\small (Left) The graph formed by connecting recurrent classes through the minimum resistance edge leaving each class, is composed of disconnected subgraphs $\Sigma_u$ with the structure illustrated above (Lemma \ref{lem:subgraphs}). All paths in $\Sigma_u$ lead to a ``head node'' $u$ that has minimal disagreement among all other nodes. Disagreement decreases along any path. (Right) The resistance tree rooted at $a$ of minimum stochastic potential. A minimum resistance edge leaving each subgraph exits via the head node (Lemma \ref{lem:head_leave}).}
	\label{fig:subgraph}
\end{figure}

\begin{lemma}\label{lem:head_leave}
	Consider the subgraph $\Sigma_u=(\mcal{A}_u,\mcal{E}_u)$ where $u\neq a$. Then there is an edge $(u,v)$ starting at the head $u$ leading to a class $v$ of another subgraph $\Sigma_{u'}$ satisfying
\begin{equation}
	(u,v) \in \argmin{a^1 \in \mcal{A}_u, a^2 \notin \mcal{A}_u}\rho_{a^1,a^2}.
\end{equation}
Furthermore. $d(u') < d(u)$.
\end{lemma}
\begin{proof}
We first observe that any edge $e_u \in \argmin{a^1 \in \mcal{A}_u, a^2 \notin \mcal{A}_u}\rho_{a^1,a^2}$ must be of type 1 or 2.  This is because the edge of type 1 out of $u$ has resistance $1+2\alpha$, so $e_u$ must incur a resistance no greater than $1+2\alpha$. This disqualifies $e_u$ to be of type 3. Type 4 is also disqualified because the only occurrence of a type 4 edge is between $u$ and $v$ for some $v\in\mcal{A}_u$. Hence, if $e_u$ is type 1, a minimal resistance path leaving $\Sigma_u$ leaves from the head $u$.

The other possibility is $e_u$ is of type 2. The head node $u$ has the most instances of $|Y|_y$ among all other nodes in $\mcal{A}_u$. This follows from the fact that all edges in $\mcal{E}_u$ except the one leaving $u$ are type 0, 1 or 2. Consequently, $u$ has the most instances of $|X|_y X|_x Y|_y$ from which a type 2 transition leaves $\Sigma_u$. Moreover, such instances appearing in any other class in $\mcal{A}_u$ also appears in $u$. Therefore, if $e_u$ is type 2, a minimal resistance path leaving $\Sigma_u$ leaves from the head $u$.

Whether $e_u$ is type 1 or 2, it leads to a class of $\Sigma_{u'}$ that either has lower disagreement or has an edge leading to another class with lower disagreement. Therefore, $d(u') < d(u)$.
\end{proof}

These results give us enough structure about the resistance trees to deduce that $a$ is the unique stochastically stable state.
 \begin{proposition}\label{a_stabilizable}
 	The profile $a$ is the unique stochastically stable state under an aggressive policy targeting $a$.
 \end{proposition}
\begin{proof}
	We need to show the rooted tree with minimal stochastic potential is rooted in $a$ (Lemma \ref{lem:stoch_potential}). Consider each subgraph $\Sigma_u$ detailed in Lemma \ref{lem:subgraphs}. By removing the edge $(u,v)$, we end up with a tree $T_u$, restricted over $\mcal{A}_u$, with minimal stochastic potential that is rooted in $u$.  Let $T$ be the rooted tree over all recurrent classes $\mcal{A}_{\text{R}}$ constructed by connecting trees $T_u$ through the minimum resistance edges leaving each head node $u \neq a$ (Lemma \ref{lem:head_leave}). Since all of these edges connect to different rooted trees with strictly lower disagreement, $T$ must be rooted in $a$, the class with minimal disagreement. 
	
	The tree $T$ has the minimum stochastic potential of all possible rooted trees because each sub-tree is of minimal potential, and edges connecting sub-trees are the minimal resistance edges leaving each sub-tree. %Altering any edges to root the tree in any class other than $a$ will result in strictly higher potential. 
\end{proof}
The structure of the minimum potential rooted tree $T$ is illustrated in Figure \ref{fig:subgraph} (right). Proposition \ref{a_stabilizable} is a sufficiency result -- the aggressive policy is a dynamic policy that stabilizes the profile $a$. Our next result asserts necessity -- the number of adversarial nodes employed by the aggressive policy is the minimum required budget to stabilize $a$.

\begin{proposition}\label{S_fewer}
	%The policy that induces $a$ as the unique stochastically stable state using the fewest number of impostors is the aggressive policy targeting $a$.
	A  policy using fewer adversarial nodes than the aggressive policy targeting $a\in\mcal{A}$ cannot stabilize $a$.
\end{proposition}
\begin{proof}
	Consider a policy $\pi'$ that uses fewer adversaries than the aggressive policy $\pi$.  Then there exists either an $x$ or $y$-segment of $a$ in which a defensive strategy is not implemented under $\pi'$ but is under $\pi$. Suppose it is an $x$ segment, and suppose $a'$ is a similar profile to $a$ in which an instance of $|X|_x$ is replaced by $|Y|_x$ on which $\pi$ applies a defensive $x$ strategy. This is the case only if $\alpha < \frac{1}{2}$ and $|L_x| \leq 1 + \left\lfloor \frac{1-\alpha}{\alpha} \right\rfloor$. Under $\pi'$, there is a path from $a$ to $a'$ using $|L_x|-1$ deviations of type $\omega(x,1,0) = \alpha$. The resistance of this path is less than $1-\alpha$. A minimum resistance path from $a'$ back to $a$ consists of at least one transition of type $\omega(y,2,x) = 1-\alpha$ or $\omega(y,2,\varnothing) = 2$. Hence, $r(a',a) \geq 1-\alpha$. Now suppose a defensive strategy fails on a $y$ segment. Suppose $a'$ is similar to $a$ in which an instance of $|Y|_y$ is replaced with $|X|_y$. A minimal path has zero resistance, since it uses only deviations of type $\omega(y,1,0) = 0$. In a similar manner, any path from $a'$ back to $a$ has resistance $r(a',a) \geq 1 + 2\alpha > 0$.
	
	In both cases, the resistance to get back to $a$ from $a'$ is greater than going from $a$ to $a'$. Therefore, the target profile $a$ cannot be stochastically stable under $\pi'$.
\end{proof}
We are now in a position to prove Theorem \ref{theorem:MI_opt}.
\begin{proof}[Proof of Theorem \ref{theorem:MI_opt}]
	We have just established that an aggressive policy targeting $a$ is the dynamic policy that stabilizes $a$ with the fewest number of adversaries. By Lemma \ref{lem:a_lengths}, a target profile must have $x$ segments of length 2 or greater and $y$ segments of length $\left\lceil \frac{2+\alpha}{1-\alpha} \right\rceil$ or greater. To implement the aggressive policy, there needs to be two adversaries for each $y$ segment of any length, and two adversaries for each $x$ segment of length no greater than $1 + \left\lfloor \frac{1-\alpha}{\alpha} \right\rfloor$ when $\alpha < \frac{1}{2}$. When $\alpha \geq \frac{1}{2}$, no adversaries are needed on any $x$ segment. Propositions \ref{a_stabilizable} and \ref{S_fewer} assert this amount of adversarial nodes are necessary and sufficient to stabilize the target profile.  By Step (2B) from Theorem \ref{theorem:SI_opt}, the minimum efficiency target profile has at most two segment patterns. We thus obtain an integer optimization problem similar to \eqref{eq:SI_opt}, except with the above constraints taken into account instead.
\end{proof}

\subsection{Proof of Theorem \ref{theorem:DU_opt}: Dynamic uninformed adversary}

In this subsection, we derive the minimum efficiency a dynamic uninformed adversary can induce on $k$-connected ring graphs. Similar to the static uninformed case (Theorem \ref{theorem:SU_opt}), the dynamic uninformed adversary  effectively can only select how many $x$ and $y$ adversaries to allocate at each time step, and cannot place them in a strategic manner. The major difference here is it can also randomize among influence sets by selecting different distributions at each time step. The idea of the proof is by being able to probabilistically attach an adversary to each agent in the network, the all $y$ profile can be stabilized independently of the budget $\gamma$. 

	Let us define $\Pi^* \subset \Pi_{\text{D}}$ as a set of dynamic uninformed policies that have the following properties. Suppose $\pi \in \Pi^*$. Then
	\begin{enumerate}[label=(\alph*)]
		\item $|S_x(t)| = 0$ for all $t=0,1,2,\ldots$.
		\item for any $i \in \mcal{N}$ and $t = 0,1,2,\ldots$, there exists a $0 \leq \tau < \infty$ such that $\text{Pr}(i \in S_y(t+\tau)) > 0$.
		\item there exists a subset $T \subset \mcal{N}$ of agents satisfying $\frac{|T|}{n} = \gamma' < \gamma$ for all $n$, such that $\text{Pr}(i \in S_y(t)) = 1$ for all $i \in T$, $t=0,1,2,\ldots$.
	\end{enumerate}
	Property (a) asserts $x$ adversaries are never utilized. Property (b) ensures any given agent is influenced by a $y$ adversary infinitely often. Property (c) says the adversary determistically influences a fixed fraction of agents in the network. In a sense, policies belonging to $\Pi^*$ are ``partially static".

\begin{lemma}
	Under a policy $\pi \in \Pi^*$, the all-$y$ and all-$x$ profiles are the only two recurrent classes.
\end{lemma}
\begin{proof}
	Under a policy $\pi\in\Pi^*$, a transition out of either $\vec{y}$ or $\vec{x}$ has a non-zero resistance. Hence, we need to show there is a zero-resistance transition out of any profile $a \notin \{\vec{y},\vec{x}\}$. First, assume $\alpha < 1/k$ (similar arguments can be applied when $\alpha > 1/k$). Profiles can be split into three categories.
	
	\noindent\textbf{1)} There is a $y$-segment of length $\geq k$. In this case, there exists a zero-resistance transition. Consider an $x$ agent that neighbors the $y$ segment of length $\geq k$. This agent could have up to $k$ neighbors also playing $x$. There is non-zero probability it becomes influenced by a $y$ adversary. It switches to $y$ with resistance $\omega(x,k,y) = [k+1 - k(1+\alpha)]_+$, which is zero if $\alpha < 1/k$.
	
	\noindent\textbf{2)} There is no $y$-segment of length $\geq k$ and there is an $x$ segment of length $\geq k$. Consider a $y$ agent that neighbors such an $x$ segment. Its payoff from playing $x$ is at least $k(1+\alpha)$, and its payoff from playing $y$ is at most $k-1$. Hence, it transitions to $x$ with zero resistance.
	
	\noindent\textbf{3)} There are no $x$ or $y$ segments of length $\geq k$. Then there always exists some agent that switches to $x$ or $y$ with zero resistance. In other words, there exists some $y$ (or $x$) agent that will attain more payoff by switching to $x$ ($y$).
\end{proof}

To determine whether $\vec{x}$ or $\vec{y}$ is stochastically stable, we now calculate the minimum resistance path between $\vec{x}$ and $\vec{y}$ ($\rho_{\vec{x}\vec{y}}$). Suppose we are in $\vec{x}$. Consider a path in which $k$ consecutive agents switch to $y$, and each agent that unilaterally switches is influenced by a $y$-adversary. This occurs with a non-zero probability due to property (b). Each unilateral switch in this sequence has a non-zero resistance. However after this sequence of $k$ switches, every subsequent switch of neighboring agents has zero resistance. The total resistance (as long as $\alpha < \frac{1}{k}$) is therefore $\rho_{\vec{x}\vec{y}} = \rho_{\vec{x}\vec{y}}^* := \sum_{i=1}^k \left[(1+\alpha)(2k-(i-1)) - i \right]$.

We can similarly calculate the minimum resistance path between $\vec{y}$ and $\vec{x}$ ($\rho_{\vec{y}\vec{x}}$). Starting from $\vec{y}$, we consider a path in which $k$ consecutive agents switch to $x$, and none of the agents are influenced by a $y$-adversary (this occurs with non-zero probability). After this sequence of $k$ switches, every subsequent switch of neighboring agents has zero resistance. Once this growing $x$ segment neighbors a member of $T$, the node belonging to $T$ will also switch with zero resistance as long as $\alpha < \frac{1}{k}$. The total resistance is therefore $\rho_{\vec{y}\vec{x}} = \rho_{\vec{y}\vec{x}}^* := \sum_{i=1}^k \left[(2k-(i-1)) - (i-1)(1+\alpha) \right]_+$.

For $\vec{y}$ to be stochastically stable, the condition $\rho_{\vec{x}\vec{y}} < \rho_{\vec{y}\vec{x}}$ must hold. This yields $\alpha < \frac{|T|}{k(|T| + k)}$. As the number of nodes tends to infinity, so does $|T|$, and we are left with the condition $\alpha < \frac{1}{k}$. This yields an efficiency $\frac{1}{1+\alpha}$ regardless of the fractional budget $\gamma$. We now proceed to show any $\pi \in \Pi_{\text{DU}} \setminus \Pi^*$ cannot induce an efficiency less than $\frac{1}{1+\alpha}$.

	\noindent\textbf{--} Suppose $\pi$ satisfies property (a) and (c), but not (b). Then there is a set $B \subset \mcal{N}$ of agents that are never influenced by $y$ impostors: $\text{Pr}(i \in S_y(t)) = 0$ for all $i\in B$ and $t=0,1,2,\ldots$. Suppose $|B| \geq k$. Suppose $B$ is one contiguous segment, and $T$ another. The collection of profiles $E\subset \mcal{A}$ that satisfy $a_B = x$ and $a_T = y$ constitutes a recurrent class. The profiles $\vec{x}$ and $\vec{y}$ are the other two recurrent classes. Similar resistance arguments show that $E$ has minimal stochastic potential, and hence $\text{LLL}(G,\alpha,\pi) = E$. The profile $a = (y_T,x_{-T})$ achieves the maximal efficiency, which yields an efficiency $> \frac{1}{1+\alpha}$ for $n$ sufficiently large. For $|B| < k$, $E$ is not a recurrent class. The same analysis for $\pi \in \Pi^*$ applies in this case, yielding an efficiency $\frac{1}{1+\alpha}$. \vspace{1mm}

	\noindent\textbf{--} If $\pi$ satisfies (a) and (b) but not (c), there are no static adversarial nodes. There are no recurrent classes other than $\vec{y}$ and $\vec{x}$. The analysis is similar to when $\pi$ did satisfy (c), where we obtain $\rho_{\vec{x},\vec{y}}^* < \rho_{\vec{y},\vec{x}}^*$ if $\alpha < 1/k$, and the opposite if $\alpha > 1/k$. Regardless, this yields an efficiency $\geq \frac{1}{1+\alpha}$. \vspace{1mm}

	\noindent\textbf{--} If $\pi$ satisfies only (a), $\vec{x}$ and $\vec{y}$ are the only recurrent classes. Regardless of which is stable, the efficiency is $\geq \frac{1}{1+\alpha}$. \vspace{1mm}
	
	%One can show $\rho_{\vec{x},\vec{y}} = \rho_{\vec{x},\vec{y}}^* + (|B|-k)k\alpha$ and $\rho_{\vec{y},\vec{x}} = \rho_{\vec{y},\vec{x}}^*$. If $|B| > \alpha^{-1}-1$, $\rho_{\vec{x},\vec{y}} > \rho_{\vec{y},\vec{x}}$, which gives an efficiency of 1. If not, then $\vec{y}$ is stochastically stable, yielding an efficiency $\frac{1}{1+\alpha}$. \vspace{1mm}

	\noindent\textbf{--} Any policy that does not satisfy property (a) induces an efficiency at least as much as a corresponding policy that does. The addition of $x$ adversaries only lessens the resistance of transitions from $y$ to $x$. Also, if there is a ``static" $x$ adversary, property (c) allows a single segment to be stable to $x$. For any number of $x$ adversaries that are ``random" in the sense of property (b), the adversary cannot stabilize alternating $x$ and $y$ segments in a repeating fashion. Therefore, no induced efficiency for large $n$ can be less than $\frac{1}{1+\alpha}$.

%%%%%%%%%%%%%%%%%%%%%%%%%%%%%%%%%%%%%%%%%%%%%%%%%%%%%%%
\section{Simulations}\label{sec:sims}
% !TEX root = main.tex
In this section, we provide numerical simulations of log-linear learning dynamics for two of the four adversarial models: static and dynamic informed (SI and DI). We verify the tightness of the lower bounds given in Theorems \ref{theorem:SI_opt} and \ref{theorem:MI_opt}. We simulate the dynamics on finite ring graphs, and compute the average efficiency the network experiences when the adversary implements an optimal policy. Our results are given in Table \ref{tab:table_sims}. We observe that the average efficiency approaches the fundamental lower bounds as the size of the graphs are increased. The experiments are set up as follows. 

For SI, we first compute the minimum efficiency SSS given $\alpha$, budget $\lfloor \gamma\cdot n\rfloor$, and network size $n$. We then determine an adversarial set $S_{xy}$ that is necessary and sufficient to stabilize the SSS (according to the proof of Theorem \ref{theorem:SI_opt}). We then run the log-linear learning dynamics initialized at the SSS and with the static $S_{xy}$. We perform a similar experiment for DI, except the aggressive policy (Definition \ref{def: agg policy}) is implemented during the dynamics. 

For the static uninformed model, we arrange a static adversarial set according to the proof of Theorem 2.1, if $\gamma \geq \alpha$. That is, it is the arrangement of solely $y$ adversaries with the given budget that induces the least amount of damage. If $\gamma < \alpha$, the $y$ adversaries are evenly spaced throughout the ring. We initialized each repetition at the stochastically stable state. For the dynamic uninformed model,  we arranged the adversary set according to the proof of Theorem 2.2. That is, it uses exclusively $y$ adversaries and randomly allocates them. We initialized each repetition with a random initial action profile.

\begin{table}[h]
    \centering
    \begin{tabular}{|c|c|c|c|}
        \hline
        \multicolumn{4}{|c|}{Static Informed Adversary}\\
        \hline
        $n$ & $\alpha = 0.3$ & $\alpha = 0.5$ & $\alpha = 0.7$ \\
        \hline
        10 & 0.6846 & 0.6667 & 1 \\
        \hline
        20 & 0.6730 & 0.6500 & 0.6529 \\
        \hline
        30 & 0.6692 & 0.6444 & 0.6588  \\
        \hline
        {} & \multicolumn{3}{|c|}{Fundamental lower bound} \\
        \hline
        {} & 0.6615 & 0.6400 & 0.6470 \\
        \hline
    \end{tabular}
    \quad
    \begin{tabular}{|c|c|c|c|}
        \hline
        \multicolumn{4}{|c|}{Dynamic Informed Adversary}\\
        \hline
        $n$ & $\alpha = 0.3$ & $\alpha = 0.5$ & $\alpha = 0.7$ \\
        \hline
        10 & 0.6384 & 0.5666 & 0.5882 \\
        \hline
        20 & 0.5730 & 0.5666 & 0.5500 \\
        \hline
        30 & 0.5948 & 0.5666 & 0.5373  \\
        \hline
        {} & \multicolumn{3}{|c|}{Fundamental lower bound} \\
        \hline
        {} & 0.5692 & 0.5416 & 0.5245 \\
        \hline
    \end{tabular} \\
    \begin{tabular}{|c|c|c|c|}
        \hline
        \multicolumn{4}{|c|}{Static Uninformed Adversary}\\
        \hline
        $n$ & $\alpha = 0.3$ & $\alpha = 0.5$ & $\alpha = 0.7$ \\
        \hline
        10 & 0.7751 & 0.6744 & 1 \\
        \hline
        20 & 0.7993 & 0.8121 & 1 \\
        \hline
        30 & 0.8371 & 0.8503 & 1  \\
        \hline
        {} & \multicolumn{3}{|c|}{Fundamental lower bound} \\
        \hline
        {} & 0.6615 & 0.6400 & 0.6470 \\
        \hline
        % FINITE n LOWER BOUNDS
        % {} & \multicolumn{3}{|c|}{Fundamental lower bound} \\
        % \hline
        % 10 & 0.7692 & 0.6667 & 1 \\
        % \hline
        % 20 & 0.7731 & 0.7667 & 1 \\
        % \hline
        % 30 & 0.8103 & 0.8222 & 1 \\
    \end{tabular}
    \quad
    \begin{tabular}{|c|c|c|c|}
        \hline
        \multicolumn{4}{|c|}{Dynamic Uninformed Adversary}\\
        \hline
        $n$ & $\alpha = 0.3$ & $\alpha = 0.5$ & $\alpha = 0.7$ \\
        \hline
        10 & 0.7999 & 0.7222 & 0.6294 \\
        \hline
        20 & 0.7769 & 0.6777 & 0.5882 \\
        \hline
        30 & 0.7692 & 0.6777 & 0.5883  \\
        \hline
        {} & \multicolumn{3}{|c|}{Fundamental lower bound} \\
        \hline
        {} & 0.7692 & 0.6666 & 0.5882 \\
        \hline
    \end{tabular}
    \caption{Simulation results of log-linear learning for the four adversarial models. Each entry represents the averaged efficiency over 30 repetitions, where each repetitions consists of $10^6$ time steps. We fix the learning parameter $\beta = 25$.  }
    \label{tab:table_sims}
\end{table}

%%%%%%%%%%%%%%%%%%%%%%%%%%%%%%%%%%%%%%%%%%%%%%%%%%%%%%%
\section{Conclusion}

This paper  investigated the susceptibility of distributed game-theoretic learning algorithms to adversarial influences. We considered a scenario of an adversary intent on maximally degrading a network system's performance guarantees associated with distributed learning algorithms. We asked 1) How susceptible are these algorithms to adversarial interference?  In particular, this paper focused on one such algorithm,  log-linear learning, that possesses nice properties in non-adversarial settings. 2) How does an adversary's sophistication and system-level knowledge impact the degradation that the adversary can do to the system? We studied both of these questions in the context of graphical coordination games. 

In particular, we considered two levels of adversarial sophistication -- static and dynamic policies -- and two levels of information, informed and agnostic about network structure. In a static policy, the adversary cannot change its influence over time. The dynamic policies we considered are stationary, in which the adversary can respond to the current system state as it evolves. While both types of policies induce asymptotic outcomes characterized by stochastically stable states, it is of interest in future work to consider non-stationary adversarial policies. That is, how can adversaries exploit dynamic policies that may not induce a stable asymptotic outcome?

An important insight gleaned from the analysis is that an adversary with a low resource budget -- described by the fraction of agents in the network it can influence -- does not benefit as much from system-level information as it would from the ability to employ dynamic strategies. On the other hand, when the adversary's budget is high, the opposite conclusion holds. While the results in this paper are adversarial-centric, these findings provide insight as to what actions a system operator could take to best protect system behavior. For instance, our analysis can inform decisions of whether to obfuscate system-level information from potential adversarial actors, or to disable capabilities of a highly sophisticated attacker.

% !TEX root = main.tex
\appendix

Here, we provide the proofs of Theorems \ref{theorem:comparison} and \ref{theorem:sat}.

\subsection{Proof of Theorem \ref{theorem:sat}}

The performance  saturates when the minimum length requirements on $x$ and $y$ segments can be stabilized. For static informed, the minimum required $y$ length is $\ell^* = \min\{\ell : \ell = \lceil\alpha(\ell+1)\rceil + 2\} = \lceil \frac{2 + \alpha}{1 - \alpha} \rceil$, and the minimum required $x$ length is 2. These lengths yield the minimum efficiency \eqref{eq:sat_val}. One needs $\ell^*$ $y$-adversaries to stabilize the $y$ segment \eqref{eq:y_nec}, and one needs $\lceil \frac{2-\alpha}{1+\alpha} \rceil$ $x$-adversaries to stabilize the $x$ segment \eqref{eq:x_nec}. Hence, the static informed adversary can achieve this pattern if and only if $\gamma \geq \frac{\ell^* + \lceil \frac{2-\alpha}{1+\alpha} \rceil}{\ell^* + 2}$.

For dynamic informed policies, Lemma \ref{lem:a_lengths} says the minimum required $x$ and $y$ segment lengths are the same as for static informed adversaries. Hence, the saturation value is the same \eqref{eq:sat_val}. However, the aggressive policy only needs 2 $y$ adversaries for each $y$ segment. It needs 2 $x$ adversaries for $x$ segments if $\alpha < 1/2$, and none if $\alpha \geq 1/2$. Hence, the performance saturates when $\gamma \geq \frac{2 + 2 \cdot \mathds{1}(\alpha < \frac{1}{2})}{\ell^* + 2}$.

\subsection{Proof of Theorem \ref{theorem:comparison}}
Our approach to proving this result is to analytically characterize $\inf_{G \in \mcal{G}^1, \pi \in \Pi_{\text{I}}(G,\gamma)} \eta(G,\alpha,\pi)$, which was expressed as an integer optimization problem in Theorem \ref{theorem:SI_opt}. Indeed, we will demonstrate this value satisfies the following relations
\begin{equation}\label{eq:inf_values}
	\begin{cases}
		= 1 - \frac{\gamma}{1+\alpha}, &\text{ if } \gamma \leq \alpha \\
		  \leq \frac{\left( \gamma - \alpha \right) (\ell^* + \alpha)}{(1+\alpha)(1-\alpha)(\ell^* + 2)}  +    \frac{\left( 1 - \gamma \right)}{(1+\alpha)(1-\alpha)}, &\text{ if } \gamma > \alpha 
	\end{cases}
\end{equation}
Comparing this value to $\frac{1}{1+\alpha}$, the minimum efficiency a dynamic uninformed adversary can induce on ring graphs, we obtain the result. In particular, it is larger when $\gamma \leq \alpha$ and smaller when $\gamma > \alpha$.

We first consider $\gamma \leq \alpha$. From Theorem 3.1, a profile of minimal efficiency consists of at most two unique segment patterns. We will show that for any profile $a$ with one segment pattern, $\eta(a) \geq 1 -  \frac{\gamma}{1+\alpha}$. Note the efficiency of $a$ can be characterized on a ring graph with just a single repetition of this segment pattern, i.e. there is only one $y$ and one $x$-segment. To stabilize the $y$-segment $L_y$, $\left\lceil \alpha(|L_y| + 1) \right\rceil + 2$ adversarial nodes are needed. This number already exceeds the adversary's fractional budget for this segment. That is, 
	\begin{equation}
		\frac{\left\lceil \alpha(|L_y| + 1) \right\rceil + 2}{|L_y|} \geq \frac{\alpha(|L_y|+1) + 2}{|L_y|} > \alpha > \gamma.
	\end{equation}
	The length $|L_x|$ of the $x$-segment must balance out the fractional budget constraint. With $n = |L_y| + |L_x|$, this means 
	\begin{equation}
		\frac{\left\lceil \alpha(|L_y| + 1) \right\rceil + 2 + \left\lceil \left[ \frac{2 - \alpha(|L_x| - 1)}{1+\alpha} \right]_+ \right\rceil}{n} \leq \gamma
	\end{equation}
	must be satisfied, from which we obtain $n  \geq \frac{ \alpha(|L_y| + 1) + 2}{\gamma}$. The efficiency of $a$ can be written as
	\begin{equation}
		\begin{aligned}
			\eta(a) &= \frac{(|L_y|-1) + (1+\alpha)(n - |L_y| - 1)}{(1+\alpha)n} \\ 
			&=  1 -   \frac{\alpha(|L_y| + 1) + 2}{(1+\alpha)n} \geq 1 - \frac{\gamma}{1+\alpha} \\
		\end{aligned}
	\end{equation}
	To show that $1 - \frac{\gamma}{1+\alpha}$ is the infimum of lower bounds on efficiency over all graphs (i.e. it is tight), consider the limit of large ring graphs, i.e. as $n \rightarrow \infty$. Suppose the adversary stabilizes one $y$-segment and one $x$-segment, and the length of the $y$-segment is a fraction $f$ of the entire network. As $n\rightarrow \infty$, the adversary needs a fractional budget $\alpha f$ to stabilize this segment, and no adversaries at all to stabilze the rest of the network to $x$. If it uses its entire budget to stabilize the $y$-segment, $\gamma = \alpha f$, it induces the efficiency
	\begin{equation}
		\lim_{n\rightarrow \infty} \frac{(|L_y|-1) + (1+\alpha)(n - |L_y| - 1)}{(1+\alpha)n} = 1 - \frac{\gamma}{1+\alpha}.
	\end{equation}
	
	Through similar arguments, one can show $\eta(a) \geq 1 - \frac{\gamma}{1+\alpha}$ holds for any profile $a$ that contains two segment patterns. We omit these details for brevity.
    
    Now, consider $\gamma > \alpha$. We will characterize a particular action profile the adversary can stabilize, whose efficiency given by the second entry of \eqref{eq:inf_values} serves as an upper bound on $\inf_{G \in \mcal{G}^1, \pi \in \Pi_{\text{I}}(G,\gamma)} \eta(G,\alpha,\pi)$. Indeed, consider an action profile where a fraction $f$ of the ring consists only of repeated segment patterns of the type described in the proof of Theorem \ref{theorem:sat}. That is, the pattern consists of the minimum required lengths on $x$ and $y$ segments. On the other hand, suppose the remaining fraction $1-f$ of the ring is a single $y$ segment. For sufficiently large graphs, a fractional budget $f$ is needed to stabilize the first portion, and another fractional budget $(1-f)\alpha$ to stabilize the second. Hence, given a total fractional budget $\gamma$, $f = \frac{\gamma - \alpha}{1-\alpha} > 0$. The fraction $f$ ranges from 0 (when $\gamma = \alpha$) to 1 (when $\gamma = 1$). The efficiency of this action profile is
    \begin{equation}
        \frac{\left( \gamma - \alpha \right) (\ell^* + \alpha)}{(1+\alpha)(1-\alpha)(\ell^* + 2)}  +    \frac{\left( 1 - \gamma \right)}{(1+\alpha)(1-\alpha)}.
    \end{equation}

\bibliographystyle{IEEEtran}
\bibliography{library}

\begin{thebibliography}{10}
\providecommand{\url}[1]{#1}
\csname url@samestyle\endcsname
\providecommand{\newblock}{\relax}
\providecommand{\bibinfo}[2]{#2}
\providecommand{\BIBentrySTDinterwordspacing}{\spaceskip=0pt\relax}
\providecommand{\BIBentryALTinterwordstretchfactor}{4}
\providecommand{\BIBentryALTinterwordspacing}{\spaceskip=\fontdimen2\font plus
\BIBentryALTinterwordstretchfactor\fontdimen3\font minus
  \fontdimen4\font\relax}
\providecommand{\BIBforeignlanguage}[2]{{%
\expandafter\ifx\csname l@#1\endcsname\relax
\typeout{** WARNING: IEEEtran.bst: No hyphenation pattern has been}%
\typeout{** loaded for the language `#1'. Using the pattern for}%
\typeout{** the default language instead.}%
\else
\language=\csname l@#1\endcsname
\fi
#2}}
\providecommand{\BIBdecl}{\relax}
\BIBdecl

\bibitem{Canty2018}
B.~Canty, P.~N. Brown, M.~Alizadeh, and J.~R. Marden, ``The impact of informed
  adversarial behavior in graphical coordination games,'' in \emph{2018 IEEE
  Conference on Decision and Control (CDC)}.\hskip 1em plus 0.5em minus
  0.4em\relax IEEE, 2018, pp. 1923--1928.

\bibitem{Martinez2007}
S.~Mart{\'{i}}nez, J.~Cort{\'{e}}s, and F.~Bullo, ``{Motion Coordination with
  Distributed Information},'' \emph{Control Systems Magazine}, vol.~27, no.~4,
  pp. 75--88, 2007.

\bibitem{Olfati-Saber2007}
R.~Olfati-Saber, J.~A. Fax, and R.~M. Murray, ``{Consensus and cooperation in
  networked multi-agent systems},'' \emph{Proceedings of the IEEE}, vol.~95,
  no.~1, pp. 215--233, 2007.

\bibitem{Zhu2009}
M.~Zhu and S.~Mart{\'{i}}nez, ``{Distributed coverage games for mobile visual
  sensors (I): Reaching the set of Nash equilibria},'' in \emph{Proceedings of
  the IEEE Conference on Decision and Control}, 2009, pp. 169--174.

\bibitem{Ramaswamy2016}
V.~Ramaswamy and J.~R. Marden, ``{A sensor coverage game with improved
  efficiency guarantees},'' in \emph{Proceedings of the American Control
  Conference}, vol. 2016-July, 2016, pp. 6399--6404.

\bibitem{Jadbabaie2003}
A.~Jadbabaie, J.~Lin, and S.~A. Morse, ``{Coordination of groups of mobile
  autonomous agents using nearest neighbor rules.}'' \emph{Transactions on
  Automatic Control}, vol.~48, no.~6, pp. 988--1001, 2003.

\bibitem{Pasqualetti_2013}
F.~{Pasqualetti}, F.~{Dorfler}, and F.~{Bullo}, ``Attack detection and
  identification in cyber-physical systems,'' \emph{IEEE Transactions on
  Automatic Control}, vol.~58, no.~11, pp. 2715--2729, Nov 2013.

\bibitem{Fawzi2014}
H.~Fawzi, P.~Tabuada, and S.~Diggavi, ``{Secure estimation and control for
  cyber-physical systems under adversarial attacks},'' \emph{IEEE Transactions
  on Automatic Control}, vol.~59, no.~6, pp. 1454--1467, 2014.

\bibitem{Sundaram_2018}
S.~{Sundaram} and B.~{Gharesifard}, ``Distributed optimization under
  adversarial nodes,'' \emph{IEEE Transactions on Automatic Control}, vol.~64,
  no.~3, pp. 1063--1076, 2019.

\bibitem{Wolpert2001}
D.~H. Wolpert and K.~Tumer, ``Optimal payoff functions for members of
  collectives,'' in \emph{Modeling Complexity in Economic and Social
  Systems}.\hskip 1em plus 0.5em minus 0.4em\relax World Scientific, 2002, pp.
  355--369.

\bibitem{McKelvey1995}
R.~D. McKelvey and T.~R. Palfrey, ``{Quantal response equilibria for normal
  form games},'' \emph{Games and Economic Behavior}, vol.~10, no.~1, pp. 6--38,
  1995.

\bibitem{Blume1995}
L.~E. Blume, ``{The Statistical Mechanics of Best-Response Strategy
  Revision},'' \emph{Games and Economic Behavior}, vol.~11, no.~2, pp.
  111--145, 1995.

\bibitem{Marden2013}
J.~R. Marden and A.~Wierman, ``{Distributed Welfare Games},'' \emph{Operations
  Research}, vol.~61, pp. 155--168, 2013.

\bibitem{Lim2013}
Y.~Lim and J.~Shamma, ``{Robustness of stochastic stability in game theoretic
  learning},'' in \emph{American Control Conference (ACC)}, 2013, pp.
  6160--6165.

\bibitem{Rahili_2014}
S.~Rahili and W.~Ren, ``Game theory control solution for sensor coverage
  problem in unknown environment,'' in \emph{53rd IEEE Conference on Decision
  and Control}.\hskip 1em plus 0.5em minus 0.4em\relax IEEE, 2014, pp.
  1173--1178.

\bibitem{Sun_2018}
C.~{Sun}, ``A time variant log-linear learning approach to the set k-cover
  problem in wireless sensor networks,'' \emph{IEEE Transactions on
  Cybernetics}, vol.~48, no.~4, pp. 1316--1325, 2018.

\bibitem{Kanakia_2016}
A.~Kanakia, B.~Touri, and N.~Correll, ``Modeling multi-robot task allocation
  with limited information as global game,'' \emph{Swarm Intelligence},
  vol.~10, no.~2, pp. 147--160, 2016.

\bibitem{Shah2010}
D.~Shah and J.~Shin, ``{Dynamics in Congestion Games},'' \emph{Proc. ACM
  SIGMETRICS'10}, vol.~38, p. 107, 2010.

\bibitem{Kearns2001}
M.~Kearns, M.~L. Littman, and S.~Singh, ``{Graphical Models for Game Theory},''
  \emph{Proceedings of the 17th conference on Uncertainty in artificial
  intelligence}, no. Owen, pp. 253--260, 2001.

\bibitem{Montanari2010}
A.~Montanari and A.~Saberi, ``{The spread of innovations in social networks.}''
  \emph{Proceedings of the National Academy of Sciences of the United States of
  America}, vol. 107, no.~47, pp. 20\,196--20\,201, 2010.

\bibitem{Young_2001}
H.~P. Young, \emph{Individual strategy and social structure: An evolutionary
  theory of institutions}.\hskip 1em plus 0.5em minus 0.4em\relax Princeton
  University Press, 2001.

\bibitem{Marden2014}
J.~R. Marden and J.~S. Shamma, ``{Game Theory and Distributed Control},'' in
  \emph{Handbook of Game Theory Vol. 4}, H.~Young and S.~Zamir, Eds.\hskip 1em
  plus 0.5em minus 0.4em\relax Elsevier Science, 2014.

\bibitem{Marden2009}
J.~R. Marden, G.~Arslan, and J.~S. Shamma, ``{Joint strategy fictitious play
  with inertia for potential games},'' \emph{IEEE Transactions on Automatic
  Control}, vol.~54, no.~2, pp. 208--220, 2009.

\bibitem{Pradelski2012}
B.~Pradelski and H.~P. Young, ``{Learning Efficient Nash Equilibria in
  Distributed Systems},'' \emph{Games and Economic Behavior}, vol.~75, no.~2,
  pp. 882--897, 2012.

\bibitem{Alos-Ferrer2010}
C.~Al{\'{o}}s-Ferrer and N.~Netzer, ``{The logit-response dynamics},''
  \emph{Games and Economic Behavior}, vol.~68, no.~2, pp. 413--427, 2010.

\bibitem{Lamport_1982}
L.~Lamport, R.~Shostak, and M.~Pease, ``The byzantine generals problem,''
  \emph{ACM Transactions on Programming Languages and Systems}, vol.~4, no.~3,
  pp. 382--401, 1982.

\bibitem{Pasqualetti_2012}
F.~{Pasqualetti}, A.~{Bicchi}, and F.~{Bullo}, ``Consensus computation in
  unreliable networks: A system theoretic approach,'' \emph{IEEE Transactions
  on Automatic Control}, vol.~57, no.~1, pp. 90--104, Jan 2012.

\bibitem{LeBlanc_2013}
H.~J. {LeBlanc}, H.~{Zhang}, X.~{Koutsoukos}, and S.~{Sundaram}, ``Resilient
  asymptotic consensus in robust networks,'' \emph{IEEE Journal on Selected
  Areas in Communications}, vol.~31, no.~4, pp. 766--781, April 2013.

\bibitem{Jaleel_2019}
H.~{Jaleel}, W.~{Abbas}, and J.~S. {Shamma}, ``Robustness of stochastic
  learning dynamics to player heterogeneity in games,'' in \emph{2019 IEEE 58th
  Conference on Decision and Control (CDC)}, Dec 2019, pp. 5002--5007.

\bibitem{Su_2016}
L.~{Su} and N.~{Vaidya}, ``Multi-agent optimization in the presence of
  byzantine adversaries: Fundamental limits,'' in \emph{2016 American Control
  Conference (ACC)}, 2016, pp. 7183--7188.

\bibitem{Paarporn_2019CDC}
K.~{Paarporn}, M.~{Alizadeh}, and J.~R. {Marden}, ``Risk and security tradeoffs
  in graphical coordination games,'' in \emph{2019 IEEE 58th Conference on
  Decision and Control (CDC)}, Dec 2019, pp. 4409--4414.

\bibitem{Borowski2015}
H.~P. Borowski and J.~R. Marden, ``{Understanding the Influence of Adversaries
  in Distributed Systems},'' in \emph{IEEE Conference on Decision and Control
  (CDC)}, 2015, pp. 2301--2306.

\bibitem{Brown2018TCNS}
P.~N. Brown, H.~P. Borowski, and J.~R. Marden, ``Security against impersonation
  attacks in distributed systems,'' \emph{IEEE Transactions on Control of
  Network Systems}, vol.~6, no.~1, pp. 440--450, 2018.

\bibitem{Monderer1996}
D.~Monderer and L.~S. Shapley, ``{Potential Games},'' \emph{Games and Economic
  Behavior}, vol.~14, no.~1, pp. 124--143, 1996.

\bibitem{Blume1993}
L.~E. Blume, ``{The Statistical Mechanics of Strategic Interaction},''
  \emph{Games and Economic Behavior}, vol.~5, no.~3, pp. 387--424, 1993.

\bibitem{Marden2012}
J.~R. Marden and J.~S. Shamma, ``{Revisiting Log-Linear Learning: Asynchrony,
  Completeness and a Payoff-based Implementation},'' \emph{Games and Economic
  Behavior}, vol.~75, no.~4, pp. 788--808, 2012.

\bibitem{Young1993}
H.~P. Young, ``{The Evolution of Conventions},'' \emph{Econometrica}, vol.~61,
  no.~1, pp. 57--84, 1993.

\end{thebibliography}

%%%%%%%%%%%%%%%%%%%%%%%%%%%%%%%%%%%%%%%%%%%%%%%%%
%=== ================APPENDIX ============================================
%%%%%%%%%%%%%%%%%%%%%%%%%%%%%%%%%%%%%%%%%%%%%%%%%

%----- Bibliography ----------------

\end{document}